\documentclass[11pt]{amsart}
\usepackage{amsmath,amsthm,amsfonts,amssymb}
\usepackage{mathtools}
\usepackage{latexsym}
\usepackage[dvipsnames]{xcolor}
\usepackage{graphicx}
\usepackage{float}
\usepackage{graphbox}
\usepackage{textcomp}
\usepackage[pagebackref,colorlinks=true,linkcolor=MidnightBlue,urlcolor=black,citecolor=RedViolet]{hyperref}
\usepackage[margin=1in]{geometry}
\usepackage{marginnote}
\usepackage{array}
\usepackage{appendix}
\usepackage{comment}

\newtheorem{theorem}{Theorem}[section]
\newtheorem{proposition}[theorem]{Proposition}
\newtheorem{lemma}[theorem]{Lemma}
\newtheorem{corollary}[theorem]{Corollary}
\newtheorem{definition}[theorem]{Definition}

\newtheorem{remark}[theorem]{Remark}
\newtheorem{example}[theorem]{Example}

\newcommand{\overbar}[1]{\mkern 1.5mu\overline{\mkern-1.5mu#1\mkern-1.5mu}\mkern 1.5mu}

\def\R{\mathbb{R}}
\def\C{\mathbb{C}}

\def\N{\mathbb{N}}

\DeclareMathOperator{\Tr}{Tr}

\DeclareMathOperator{\E}{\mathbb{E}}
\DeclareMathOperator{\Var}{Var}
\DeclareMathOperator{\Prob}{P}

\begin{document}

\title{Moment-based PPT criteria for random bipartite states}

\author{Cécilia Lancien}
\address[Cécilia Lancien]{Université Grenoble Alpes, CNRS, Institut Fourier, France}
\email{cecilia.lancien@univ-grenoble-alpes.fr}
\author{Kieran McShane}
\address[Kieran McShane]{Université Grenoble Alpes, CNRS, Institut Fourier, France}
\email{kieran.mcshane@univ-grenoble-alpes.fr}

\begin{abstract}
Moment-based relaxations of the positive partial transpose (PPT) criterion have been recently introduced, as a hierarchy of entanglement criteria involving only experimentally accessible quantities of a given bipartite state. The goal of this work is to study their typical detection performance on high-dimensional bipartite systems. Concretely, we investigate whether random bipartite mixed states on $\mathbb C^d\otimes\mathbb C^d$, obtained as the marginal over an environment $\mathbb C^s$ of a uniformly distributed pure state, generically satisfy or violate them. For each fixed level $m\in\mathbb N$ in this hierarchy of moment-based PPT criteria, we are able to identify a threshold environment dimension $s=\lambda_md^2$ at which the behavior of the associated random state switches from violating to satisfying it, with probability going to $1$ as $d$ grows. The proof combines combinatorics of permutations techniques to estimate the average value of moments of partially transposed random states and concentration of measure arguments to bound the probability of deviating from such average, when the underlying local dimension $d$ is large. We additionally need tools from the theory of Hankel determinant evaluation via orthogonal polynomials.
\end{abstract}

\date{\today}

\maketitle

\vspace{-1.2em}
\setcounter{tocdepth}{2}
{\small \tableofcontents}

\newpage

\section{Introduction}

Entanglement is a key concept in quantum information theory, yet characterizing it is a notoriously difficult task. Indeed, deciding whether an arbitrary multipartite quantum state $\rho$ is entangled or separable is known to be in general a computationally hard problem \cite{Gurvits2003, Gharibian2010}. This challenge justifies the use of entanglement criteria that are easier to check than separability itself and can act as powerful witnesses for entanglement. The most widely used of these, at least in the bipartite case, is the positive under partial transpose (PPT) criterion \cite{Peres1996,Horodecki1996}. While simple from a mathematical perspective, its practical application is demanding, as it requires the knowledge of the full spectrum of the partially transposed state $\rho^\Gamma$, which in turn requires full tomography of the state $\rho$. A practical solution is found in a further relaxation of the PPT criterion, which involves only experimentally measurable quantities, namely the low-order moments of $\rho^\Gamma$, i.e.~$p_k = \Tr((\rho^\Gamma)^k)$ for $k\in\N$ up to some fixed value, instead of its entire spectrum \cite{Elben2020}.

The goal of this paper is to investigate whether a random high-dimensional bipartite state typically satisfies or violates these moment-based PPT criteria. This can be seen as a more systematic approach to an exploration that was initiated in \cite{Car24}. We analyze a natural model of random mixed states arising as the partial trace of a uniformly distributed pure state in a larger system, where a parameter $\lambda > 0$ quantifies the mixedness of the resulting state. Our main result is the identification, for each level in this hierarchy of moment-based PPT criteria, of a sharp threshold for the parameter $\lambda$ at which the behavior of a random state switches from typically violating to typically satisfying the criterion. For the first (non-trivial) level in this hierarchy, the so-called $p_3$-PPT criterion, such threshold value of $\lambda$ is equal to $1$. While as the level increases, it goes to $4$, which is the known threshold for violating vs satisfying the PPT criterion \cite{Aubrun12}. Our results thus provide a quantitative comparison between the performance of each of these implementable entanglement criteria and that of the less practical PPT criterion, on typical instances. We give below an informal version of our main result, whose precise statement appears in Corollary \ref{cor:gen-threshold}.

\begin{theorem}[Main result, informal version]
Let $\rho$ be the reduced state on $\mathbb C^d\otimes\mathbb C^d$ of a uniformly distributed pure state on $\mathbb C^d\otimes\mathbb C^d\otimes\mathbb C^s$, with $s=\lambda d^2$. Given $m\in\mathbb N$, define
\[
\lambda_m=4\cos^2\left(\frac{\pi}{m+2}\right).
\]
If $\lambda<\lambda_m$, resp.~$\lambda>\lambda_m$, then $\rho$ violates, resp.~satisfies, with high probability the $m^{\text{th}}$ level in the hierarchy of moment-based PPT criteria. %$p_{2m+1}$-PPT criterion with high probability. 
\end{theorem}

The $m^{\text{th}}$ level in the hierarchy of moment-based PPT criteria, which is usually referred to as the $p_{2m+1}$-PPT criterion, is properly described in Definition \ref{def:p_k-PPT}. We omit such formal definition for now in order to keep the introduction light. Also, as we will see later on, when we say that a property of a random state holds `with high probability', we mean 'with probability going to $1$ as the underlying local dimension $d$ grows'.

Our analysis is grounded in the methods of random matrix theory and asymptotic geometric analysis. The mathematical techniques involved include the use of Wick's formula and combinatorics of permutations to compute the average and variance of the moments of partially transposed Wishart matrices, and to identify their asymptotic scaling (i.e.~as $d\to\infty$). We then combine these estimates with concentration of measure inequalities for polynomials in Gaussian variables to bound the probability of deviating from the asymptotic average, for $d$ being large but finite. We also crucially need at some point to obtain a closed-form expression for the Hankel determinants associated with the sequence of asymptotic average partially transposed moments, which is done through orthogonal polynomial methods.

The paper is organized as follows. Section \ref{sec:p_k-PPT} provides the necessary background on the moment-based PPT criteria for bipartite states. In Section \ref{sec:random-states}, we formally define the random state model that we look at and review previously established thresholds for their entanglement-related properties. Section \ref{sec:average} contains our first main technical results, namely estimates of the average and variance of the moments of partially transposed Wishart matrices. We then recall in Section \ref{sec:concentration} known deviation inequalities for polynomials in Gaussian variables and apply them to prove that the moments, and hence the associated Hankel determinants, we are interested in concentrate around their average. Combined with the derivation of a closed-form expression for the latter quantities, this allows us to establish in Section \ref{sec:general_threshold} our threshold result for the typical violation or satisfaction of moment-based PPT criteria. Finally, we discuss in Section \ref{sec:unbalanced} how our results generalize to the regime where the two local dimensions are not the same. We conclude in Section \ref{sec:conclusion} with a discussion of potential avenues for future work. Appendix \ref{app:second_approach} gathers an alternative approach (to the one presented in Section \ref{sec:concentration}) to proving concentration of moments, based on estimating local Lipschitz constants of the functions we are dealing with, which gives slightly better results than the one presented in the main text but is considerably more tedious. Finally, Appendix \ref{app:moment} contains one technical result (namely a moment recurrence relation) needed for computing the asymptotic Hankel determinants in Section \ref{sec:general_threshold}. 

\section{Hierarchy of moment-based PPT criteria for entanglement}
\label{sec:p_k-PPT}

\subsection{Brief reminder about entanglement detection} \hfill\vspace{0.1cm}

A quantum state on a Hilbert space $\mathcal H$ is described by an operator $\rho$ on $\mathcal H$ that is Hermitian positive semi-definite (PSD) and has trace $1$. The extreme points of the set of quantum states on $\mathcal H$ are the so-called pure states, which are simply rank-$1$ projections, i.e.~of the form $\rho=xx^*$ for
$x\in\mathcal H$ a unit vector. General states on $\mathcal H$, also called mixed states, are convex combinations of pure states.

If $\mathcal H$ is bipartite, i.e.~$\mathcal H=\mathcal{H}_1 \otimes \mathcal{H}_2$ for some Hilbert spaces $\mathcal H_1,\mathcal H_2$, an important distinction is that of separable vs entangled states. A state $\rho$ on $\mathcal{H}_1 \otimes \mathcal{H}_2$
is called separable if it can be written as a convex combination of product states, i.e.~if
\[ \rho=\sum_{i=1}^r \lambda_i \rho_1^{(i)}\otimes\rho_2^{(i)}, \]
for some $r\in\mathbb N$, $\lambda_1,\ldots,\lambda_r\geq 0$ such that
$\lambda_1+\cdots+\lambda_r=1$ and states $\rho_1^{(1)},\ldots,\rho_1^{(r)}$, resp.~$\rho_2^{(1)},\ldots,\rho_2^{(r)}$, on $\mathcal H_1$, resp.~$\mathcal H_2$. And it is called entangled if it cannot be written in such a way. Only systems that are in an entangled state exhibit genuinely quantum correlations between their subsystems, and can thus provide an advantage over classical systems in many information processing tasks. Being able to characterize entanglement is thus crucial in practice. However, this problem (and even approximate versions of it) is known to be NP-hard \cite{Gurvits2003,Gharibian2010}.

To address this issue, a solution is to design necessary conditions for separability
that are easier to check than separability itself, aka entanglement criteria. The most
famous such criterion is the so-called positive under partial transpose (PPT) criterion
\cite{Peres1996, Horodecki1996}. It consists in the obvious observation that, if a state
$\rho$ on $\mathcal{H}_1\otimes\mathcal{H}_2$ is separable, then its partial
transposition, say with respect to the second subsystem, is PSD, i.e.
\[  \rho^\Gamma:=\left[\mathrm{id}_{\mathcal H_1}\otimes\mathrm{T}_{\mathcal H_2}\right](\rho) \geq 0 . \]
Here, $\mathrm{id}_{\mathcal H_1}$, resp.~$\mathrm T_{\mathcal H_2}$, denotes the identity, resp.~transposition, map on the space of operators on $\mathcal H_1$, resp.~$\mathcal H_2$. Being PPT is thus a
necessary condition for separability, which is not sufficient in general (except if one subsystem has dimension $2$ and the other one has dimension $2$ or $3$). This means that there exist entangled states which nevertheless satisfy the PPT condition. 

Checking whether a given state $\rho$ is PPT is mathematically and computationally easy, as it only requires checking whether the matrix $\rho^\Gamma$ is PSD. Nevertheless, from an experimental point of view, accessing the full spectrum of $\rho^\Gamma$ requires full state tomography of $\rho$, and hence a number of measurements that is quadratic in the dimension of $\mathcal H$ \cite{ODonnell2015,Haah2015}. This observation motivates the search for further relaxations of the PPT criterion that would involve only
quantities that are more efficiently accessible experimentally, such as e.g.~the low order moments of $\rho^\Gamma$. Indeed, one can design practically implementable protocols (based on local randomized measurements and classical shadows) that approximate $\Tr((\rho^\Gamma)^k)$ for small $k\in\mathbb N$ with a number of measurements that is only linear in the dimension of $\mathcal H$ \cite{Elben2020}.

\subsection{Hankel matrices of partially transposed moments and $p_{2m+1}$-PPT criterion} \hfill\vspace{0.1cm}

Let us now explain more concretely how the PPT property of a state $\rho$ on $\mathcal H=\mathcal H_1\otimes\mathcal H_2$ can be investigated using the moments of its partial transpose. 
For $k\in\mathbb N$, the $k$-th moment of $\rho^\Gamma$ is defined as
\[ p_k=p_k(\rho):=\Tr\left(\left(\rho^\Gamma\right)^k\right). \]
We have $p_0=\Tr(I)$ and $p_1=\Tr(\rho^\Gamma)=\Tr(\rho)=1$. Since $\rho^\Gamma$ is Hermitian but not necessarily PSD, the moments $p_k$ are real but not necessarily non-negative for $k>1$.

Given $m\in\mathbb N$, there are two natural Hankel matrices of size $m+1$ associated with this moment sequence. The unshifted Hankel
matrix, which involves the moments $p_0,\ldots,p_{2m}$, is
\[
    A_m\left(\rho^\Gamma\right)
    :=
    \left(p_{i+j}(\rho)\right)_{0\leq i,j\leq m}
    =
    \begin{pmatrix}
        p_0 & p_1 & p_2 & \cdots & p_m \\
        p_1 & p_2 & p_3 & \cdots & p_{m+1} \\
        p_2 & p_3 & p_4 & \cdots & p_{m+2} \\
        \vdots & \vdots & \vdots & \ddots & \vdots \\
        p_m & p_{m+1} & p_{m+2} & \cdots & p_{2m}
    \end{pmatrix},
\]
while the shifted Hankel matrix, which involves the moments $p_1,\ldots,p_{2m+1}$, is
\begin{equation}\label{eq:def-B_k}
    B_m\left(\rho^\Gamma\right) := \left(p_{i+j+1}(\rho)\right)_{0\leq i,j\leq m} =
    \begin{pmatrix}
        p_1 & p_2 & p_3 & \cdots & p_{m+1} \\
        p_2 & p_3 & p_4 & \cdots & p_{m+2} \\
        p_3 & p_4 & p_5 & \cdots & p_{m+3} \\
        \vdots & \vdots & \vdots & \ddots & \vdots \\
        p_{m+1} & p_{m+2} & p_{m+3} & \cdots & p_{2m+1}
    \end{pmatrix}.
\end{equation}

Let us first explain why this distinction is in fact essential. Given $a=(a_0,\ldots,a_m)\in\mathbb R^{m+1}$, define the real polynomial $q(x)=\sum_{i=0}^m a_i x^i$. Then, for any Hermitian operator $M$ on $\mathcal H$, we have
\[ a^T A_m(M)a = \sum_{i,j=0}^m a_i a_jp_{i+j} = \Tr\left(q(M)^2\right)\geq 0. \]
Thus $A_m(M)$ carries no information about whether $M$ is PSD. In contrast, we have
\[ a^T B_m(M)a = \sum_{i,j=0}^m a_i a_jp_{i+j+1} = \Tr\left(Mq(M)^2\right), \]
which might detect whether $M$ is PSD or not. 

This brings us to the following definition, which serves as starting point for this work. It was first introduced in \cite{YIG}.

\begin{definition} \label{def:p_k-PPT}
For $m\geq0$, we say that a state $\rho$ on $\mathcal H_1\otimes\mathcal H_2$ satisfies the $p_{2m+1}$-PPT criterion if
\[ B_m\left(\rho^\Gamma\right)\geq 0. \]
Since $B_\ell$ is a principal sub-matrix of $B_m$ whenever $\ell\leq m$, these criteria form a nested hierarchy. What is more, the level $m$ uses only the first $2m+1$ moments of $\rho^\Gamma$.
\end{definition}

\begin{remark} \label{rem:normalization}
    Let us make here a very simple observation, that will actually turn out to be particularly useful to us later on. Given a state $\rho$ on $\mathcal H_1\otimes\mathcal H_2$ and $\alpha>0$, set $\rho_\alpha:=\alpha\rho$. Then, checking that $B_m(\rho^\Gamma)\geq 0$ is equivalent to checking that $B_m(\rho_\alpha^\Gamma)\geq 0$. Indeed, since for all $k\in\mathbb N$, $\Tr((\rho_\alpha^\Gamma)^k)=\alpha^k\Tr((\rho^\Gamma)^k)$, it is easy to see that, setting 
    \[ D_\alpha:=\mathrm{diag}\left(\alpha^{1/2},\alpha^{3/2},\ldots,\alpha^{m+1/2}\right), \]
    we have
    \[ B_m\left(\rho_\alpha^\Gamma\right) = D_\alpha B_m\left(\rho^\Gamma\right) D_\alpha^T. \]
    Hence, since $D_\alpha$ is clearly invertible, we have $B_m(\rho_\alpha^\Gamma)\geq 0$ if and only if $B_m(\rho^\Gamma)\geq 0$. This shows that, in order to check that $\rho$ satisfies the $p_{2m+1}$-PPT criterion, one can in fact check the positive semi-definiteness of the Hankel matrix $B_m$ associated to the partial transposition of any positive multiple of $\rho$.
\end{remark}

The connection between satisfying the $p_{2m+1}$-PPT criterion and satisfying the actual PPT criterion is made precise by the following result.

\begin{proposition}\label{prop:moment_ppt_equivalence}
Given $M$ a Hermitian operator acting on a finite-dimensional Hilbert space $\mathcal H$, we have
\[
    M\geq 0
    \quad\Longleftrightarrow\quad
    B_m(M)\geq 0 \text{ for all } m\in\mathbb N.
\]
What is more, if $\mathcal H$ is of dimension $D$, it is enough to check the conditions up to $m=D-1$.
\end{proposition}

\begin{proof}
Given $m\in\mathbb N$ and $a=(a_0,\ldots,a_m)\in\mathbb R^{m+1}$, set
\[ q(x)=\sum_{i=0}^m a_i x^i. \]
By the definition of $B_m(M)$, we have
\[ a^T B_m(M)a = \sum_{i,j=0}^m a_i a_j \Tr\left(M^{i+j+1}\right) = \Tr\left(Mq(M)^2\right). \]
If $M\geq0$, since $q(M)^2\geq 0$ by definition, we have $\Tr(Mq(M)^2)\geq 0$. We have thus shown that, for all $a\in\mathbb R^{m+1}$, $a^T B_m(M)a\geq 0$, which means that $B_m(M)$ is PSD.

Conversely, suppose that $M$ has a negative eigenvalue $\lambda_1<0$. Let
$\lambda_1,\ldots,\lambda_r$ be the distinct eigenvalues of $M$, with multiplicities $m_1,\ldots,m_r$. By interpolation, there exists a real polynomial $q$ of degree at most $r-1$ such that $q(\lambda_1)=1$ and $q(\lambda_j)=0$ for all $2\leq j\leq r$. Then, for such $q$,
\[ \Tr\left(Mq(M)^2\right)=m_1\lambda_1 <0. \]
Thus by the same reasoning as before, there exists $a=(a_0,\ldots,a_{r-1})\in\mathbb R^{r}$ such that $a^T B_{r-1}(M)a<0$, which means that $B_{r-1}(M)$ is not PSD. Since $r\leq D$, the conditions up to $m=D-1$ already detect any negative eigenvalue.
\end{proof}

As an immediate consequence of Proposition \ref{prop:moment_ppt_equivalence} we have that a state $\rho$ on $\mathbb C^d\otimes\mathbb C^d$ is PPT if and only if it satisfies the $p_{2m+1}$-PPT criteria for all $1\leq m\leq d^2-1$. This provides a hierarchy of relaxations of the PPT criterion which collapses to the actual PPT criterion at a finite level. It is naturally indexed by odd moments due to the observation that shifted rather than unshifted Hankel matrices are the relevant objects for detecting non-PSD matrices.

\begin{example}
As a concrete example, let us look at the simplest non-trivial criterion in the hierarchy, namely the $p_3$-PPT criterion. It requires to check the positive semi-definiteness of
\[ B_1 = \begin{pmatrix} p_1 & p_2 \\ p_2 & p_3 \end{pmatrix}. \]
$B_1$ being PSD is equivalent to its principal minors being non-negative. The $1 \times 1$ minor condition is $p_1 \geq 0$. As we have already explained, $p_1=1$ for any quantum state, so this is always satisfied. The $2 \times 2$ minor condition is $\det(B_1) \geq 0$, i.e.
\[ p_1 p_3 - p_2^2 \geq 0. \]
This provides a concrete necessary condition based on the first three moments of the partial transpose that must hold for any PPT (and thus any separable) state.
\end{example}

\section{Random bipartite mixed states and their entanglement-related properties}
\label{sec:random-states}

\subsection{Model of random induced states} \hfill\vspace{0.1cm}

We now properly introduce the model of random bipartite mixed states that we will consider in this work. The definitions and results presented below can be found for instance in \cite{Zyczkowski_2001} and \cite[Section 6.2.3]{AubrunSzarek2017}. They form the foundation for understanding random states.

In what follows, we denote by $\mu_D$ the uniform distribution on the unit sphere (for the Euclidean norm) of $\C^D$. It is characterized by the property of being the unique probability measure on the unit sphere of $\C^D$ such that, if $x\sim\mu_D$, then $Ux\sim\mu_D$ for any unitary operator $U$ on $\C^D$. We also recall the following useful alternative way of describing $\mu_D$: if $g\in\C^D$ is a standard Gaussian vector (i.e.~its entries are independent complex Gaussian variables with mean $0$ and variance $1$), then $g/\|g\|\sim\mu_D$.

\begin{definition}[Random induced states]
Given $D,s\in\mathbb N$, let $x \in \C^D \otimes \C^s$ be a uniformly distributed unit vector and consider $xx^*$ the associated random pure state on $\C^D\otimes\C^s$. The random mixed state $\rho$ on $\C^D$ defined as
\[ \rho = \left[\mathrm{id}_D \otimes \Tr_s\right](xx^*) \]
is said to be induced by the environment $\C^s$. Here $\mathrm{id}_D$, resp.~$\Tr_s$, denotes the identity on $\mathbb C^D$, resp.~the trace on $\mathbb C^s$. Its distribution will be denoted by $\nu_{D,s}$ in what follows.
\end{definition}

Note that, if $s\leq D$, then by construction $\rho\sim\nu_{D,s}$ has rank at most $s$. More precisely, $\rho$ is uniformly distributed (i.e.~according to the induced Lebesgue measure) on the set of mixed states of $\C^D$ of rank at most $s$. In particular, $\nu_{D,D}$ is nothing else than the uniform measure on the set of mixed states of
$\C^D$.

It is easy to deduce a Gaussian description of random induced states from that of uniformly distributed unit vectors. Let $G \in M_{D,s}(\C)$ be a standard Gaussian matrix (i.e.~its entries are independent complex Gaussian variables with mean $0$ and
variance $1$) and define the PSD matrix $W\in M_D(\C)$ by
\[ W = GG^*. \]
Its distribution is called the Wishart distribution with parameters $(D, s)$ and is denoted by $\mathcal W_{D,s}$. The following result connects $\nu_{D,s}$ and $\mathcal W_{D,s}$.

\begin{proposition} \label{prop:induced-state-Wishart}
Given $D,s\in\mathbb N$, let $W\sim\mathcal W_{D,s}$. Then, $W/\Tr(W)\sim\nu_{D,s}$.
What is more, the random variables $W/\Tr(W)$ and $\Tr(W)$ are independent.
\end{proposition}

An equivalent way of phrasing Proposition \ref{prop:induced-state-Wishart} is as follows: letting $G \in M_{D,s}(\C)$ be a standard Gaussian matrix and setting $X=G/\|G\|_2$ (where $\|\cdot\|_2$ stands for the Hilbert--Schmidt or Schatten-$2$ norm), $XX^*\sim\nu_{D,s}$. Indeed, if $W=GG^*$, then $\Tr(W)=\|G\|_2^2$, so we do have $XX^*=W/\Tr(W)$. This is interesting because such $X$ is uniformly distributed on the unit sphere for the Hilbert--Schmidt norm of $M_{D,s}(\C)$. So this shows that an alternative way of viewing $\rho\sim\nu_{D,s}$ is as $\rho=XX^*$ for such uniformly
distributed unit matrix $X\in M_{D,s}(\C)$.

Throughout this work, we will consider a balanced bipartite quantum system $\C^d \otimes \C^d$, so that $D=d^2$, coupled to an environment $\C^s$. We moreover focus on the case where the local system dimension $d$ is large and where the environment dimension $s$
scales as $s = s_d \sim \lambda d^2$, where $\lambda > 0$ is a constant parameter. As we will see, this scaling is the appropriate one for studying the entanglement-related properties that we are interested in.

\subsection{Known threshold results for entanglement-related properties} \hfill\vspace{0.1cm}

A general question that arises when studying such random induced states on a bipartite quantum system is the following: Given an entanglement criterion, is there a threshold $s_d$ for $s$ at which the behavior of a random states $\rho$ on $\C^d\otimes\C^d$, induced by an environment $\C^s$, transitions from typically violating it to typically satisfying it? The following result, established in \cite{Aubrun12}, provides a positive answer for the PPT criterion. It shows that $\rho$ is typically not PPT for $s<4d^2$ and typically PPT for $s>4d^2$.

\begin{theorem}
Let $d,s\in\mathbb N$ and let $\rho$ be a random state on $\C^d\otimes\C^d$, induced by an environment $\C^s$. For every $\epsilon > 0$, there exist $c(\epsilon),C(\epsilon)>0$ such that
\begin{itemize}
    \item[(1)] if $s \leq (4 - \epsilon)d^2$, then
    \[
    \Prob(\rho \text{ is not PPT}) \geq 1-C(\epsilon) e^{-c(\epsilon)s};
    \]
    \item[(2)] if $s\geq (4 + \epsilon)d^2$, then
    \[
    \Prob(\rho \text{ is PPT}) \geq 1 - C(\epsilon) e^{-c(\epsilon)s}.
    \]
\end{itemize}
\end{theorem}

Let us point out that, for all most studied entanglement criteria, such sharp transition phenomenon occurs, at a threshold environment dimension $s_d$ that scales as $d^2$. This is for instance the case for the realignment criterion \cite{AN12} or the extendibility
criterion \cite{Lan16}. On the other hand, the threshold for switching from separability to entanglement happens at a much higher environment dimension $s_d$, of order $d^3$ \cite{ASY14}. This illustrates the fact that, on high dimensional bipartite systems, all
entanglement criteria actually perform quite poorly, in the sense that a lot of entangled states are not detected.

Our goal in what follows is to identify, for any fixed $m\in\mathbb N$, what is the threshold dimension $s_d(m)$ at which the typical behavior of a random induced state $\rho$ on $\C^d\otimes\C^d$ switches from violating to satisfying the $p_{2m+1}$-PPT criterion. In order to do this, we need to estimate the value of $\Tr((\rho^\Gamma)^k)$ for all $1\leq k\leq 2m+1$ and understand whether the Hankel matrix $B_m$ associated to such moments is PSD. As explained in Remark \ref{rem:normalization}, one can equivalently do this with any positive multiple of $\rho$ instead of $\rho$, for instance with $W$ the Wishart matrix of parameters $(d^2,s)$ such that $\rho=W/\Tr(W)$. This is what we proceed to do in the next sections.

\section{Expectation and variance of moments of partially transposed Wishart matrices} \label{sec:average}

Our goal in this section is to compute the expectation and variance of moments of a partially transposed Wishart matrix $W\sim\mathcal W_{d^2,s}$, and estimate their asymptotic scaling (as $d,s\to\infty$ with $s$ of order $d^2$). In order to do this, we first need to recall some well-known facts about Gaussian variables and combinatorics of permutations.

\subsection{Wick's formula for the expectation of products of Gaussian variables} \hfill\vspace{0.1cm}

In what follows, we say that $Z$ is a complex Gaussian variable with mean $0$ and variance $1$, which we denote by $Z \sim \mathcal{N}_{\C}(0,1)$, if $Z=(X+iY)/\sqrt{2}$ with $X,Y$ independent real Gaussian variables with mean $0$ and variance $1$. Recall that for $Z \sim \mathcal{N}_{\C}(0,1)$, we have
$\E(Z)=0$, $\E(Z^2)=0$ and $\E(|Z|^2)=1$. Wick's theorem provides a formula for the expectation of products of Gaussian variables \cite{Zvon97}. For our purposes, we need the complex version involving pairs of variables and their conjugates, whose joint
expectation can be computed via pairings.

\begin{theorem}[Wick's formula for complex Gaussian variables]
\label{thm:wick_complex_general}
Let $\{ Z_t \}_{t=1}^k$ and $\{ Z_t' \}_{t=1}^k$ be two sets of complex centered jointly Gaussian random variables. Then the expectation of the product of these variables and their conjugates is given by the sum over all pairings between the set $\{Z_t\}_{t=1}^k$ and the set $\{\overbar{Z}_t'\}_{t=1}^k$. Concretely, we have
\[
\E \left[ Z_1 \dots Z_k \overbar{Z}_1' \dots \overbar{Z}_k' \right] = \sum_{\pi \in S_k} \prod_{t=1}^k \E \left[ Z_t \overbar{Z}_{\pi(t)}' \right]
\]
where $S_k$ is the symmetric group on $\{1, \dots, k\}$.
\end{theorem}

In order to compute moments of $W=GG^*$, we need to be able to compute the expectation of quantities of the form $\prod_{t=1}^k Z_t \overbar{Z}_t'$, where $Z_t$ and $Z_t'$ are entries of $G$. By Theorem \ref{thm:wick_complex_general}, this expectation is non-zero
only if the variables can be perfectly paired. Specifically, for
independent $\mathcal{N}_{\C}(0,1)$ entries $G_{\alpha\beta}$, we have
\[
\E \left[
    \prod_{t=1}^k G_{\alpha_t, \beta_t}
    \overbar{G}_{\alpha'_t, \beta'_t}
\right]
= \sum_{\pi \in S_k} \prod_{t=1}^k
    \E \left[
        G_{\alpha_t, \beta_t}
        \overbar{G}_{\alpha'_{\pi(t)}, \beta'_{\pi(t)}}
    \right]
= \sum_{\pi \in S_k} \prod_{t=1}^k
    \delta_{\alpha_t, \alpha'_{\pi(t)}}
    \delta_{\beta_t, \beta'_{\pi(t)}} .
\]
where $S_k$ is the symmetric group on $\{1, \dots, k\}$ and $\delta_{a,b}$ is the Kronecker delta.

\subsection{Combinatorics of permutations and non-crossing partitions} \label{sec:combinatorics} \hfill\vspace{0.1cm}

As we have already seen a first glimpse of in Theorem \ref{thm:wick_complex_general}, permutations play a crucial role when studying the expectation of products of Gaussian random variables. We therefore gather here some background on the combinatorics of permutations that will be useful to us later on. We refer the reader to \cite[Lectures 9 and 23]{Nica_Speicher_2006} for more details on the definitions and results presented below.

Let $S_k$ denote the symmetric group on $\{1, \dots, k\}$ (i.e.~the set of permutations of $k$ elements). Given $\sigma
\in S_k$, we denote by
\begin{itemize}
    \item $\#(\sigma)$ the number of cycles in its disjoint cycle decomposition (including fixed points);
    \item $|\sigma|$ its length, i.e.~the minimal number $q \in \N$ such that $\sigma$ can be written as a product of $q$ transpositions.
\end{itemize}
These parameters satisfy the relation
    \begin{equation} \label{eq:length_cycles}
        \#(\sigma)+|\sigma| = k.
    \end{equation}

What is more, we can define a distance $\mathrm d$ on $S_k$ by, for all $\sigma_1,\sigma_2\in S_k$,
\[
\mathrm d(\sigma_1,\sigma_2) = \left|\sigma_1^{-1}\sigma_2\right|.
\]
By the triangle inequality for $\mathrm d$ we thus have that, for any
$\sigma_1,\sigma_2,\sigma_3\in S_k$,
\[
\left|\sigma_1^{-1}\sigma_2\right| + \left|\sigma_2^{-1}\sigma_3\right| \geq \left|\sigma_1^{-1}\sigma_3\right|, \]
or equivalently, using relation \eqref{eq:length_cycles},
\begin{equation} \label{eq:triangle-inequality}
\#\left(\sigma_1^{-1}\sigma_2\right) + \#\left(\sigma_2^{-1}\sigma_3\right) \leq k + \#\left(\sigma_1^{-1}\sigma_3\right).
\end{equation}
If there is equality in the above inequality, we say that $\sigma_2$ is on the geodesic between $\sigma_1$ and $\sigma_3$. What is more, for fixed $\sigma_1$ and $\sigma_3$, the function $\sigma\in S_k\mapsto \#(\sigma_1^{-1}\sigma) + \#(\sigma^{-1}\sigma_3)$ can only take values $k+\#(\sigma_1^{-1}\sigma_3)-2l$, for some $l\in\N$. This can be shown recursively by observing that multiplying $\sigma$ by a transposition either increases or decreases by $1$ each of the quantities $\#(\sigma_1^{-1}\sigma)$ and $\#(\sigma^{-1}\sigma_3)$.

A partition $\pi$ of the set $\{1, \dots, k\}$ is a collection $\{V_1, \ldots, V_r\}$ of non-empty, mutually disjoint subsets of $\{1, \dots, k\}$, called blocks, such that $V_1\cup\cdots\cup V_r = \{1, \dots, k\}$. Given $1\leq i,j\leq k$, we write $i \sim_\pi j$ if $i$ and $j$ belong to the same block of $\pi$. The number of blocks $r$ is denoted by $\#(\pi)$.

A partition $\pi$ is called non-crossing if there do not exist $1\leq i_1 < i_2 < j_1 < j_2 \leq k$ such that $i_1 \sim_\pi j_1$ and $i_2 \sim_\pi j_2$ but $i_1\not\sim_\pi i_2$. The set of all non-crossing partitions of $\{1, \dots, k\}$ is denoted by $NC(k)$. Its cardinal is the so-called $k^{\text{th}}$ Catalan number, which
we denote $\mathrm{Cat}(k)$ and whose value is
\[
\mathrm{Cat}(k) = \frac{1}{k+1}\binom{2k}{k}.
\]
We also mention that the set of non-crossing partitions of $\{1, \dots, k\}$ is in bijection with the set of non-crossing pairings (i.e.~partitions where each block has cardinal $2$) of $\{1, \dots, 2k\}$, which we denote by $NC_2(2k)$. We thus have
\[
\mathrm{Cat}(k) = |NC(k)| = |NC_2(2k)|.
\]

In what follows, we denote by $\gamma\in S_k$ the $k$-cycle permutation $\gamma=(k,k-1,\ldots,2,1)$. If we apply inequality \eqref{eq:triangle-inequality} to the particular case where $\sigma_1=\gamma$ and $\sigma_3=\mathrm{id}$, we get that, for all $\sigma\in S_k$,
\begin{equation} \label{eq:triangle-inequality-gamma-id}
\#(\gamma^{-1}\sigma) + \#(\sigma) \leq k + 1.
\end{equation}
It turns out that permutations that are on the geodesic between $\gamma$ and $\mathrm{id}$ (i.e.~for which the above inequality is an equality) are exactly those whose associated cycle partition is non-crossing, with the ordering inside each cycle being the same as the ordering inside the cycle of $\gamma$. Those are sometimes called by extension non-crossing permutations, and we make the slight abuse of denoting this subset of $S_k$ by $NC(k)$ as well.

\begin{remark} 
As we will see below, moment formulas involving Gaussian matrices, obtained by making use of Wick's formula, usually involve sums over the full set of permutations $S_k$. On the contrary, their asymptotic behavior is often captured by a restricted sum over the subset of non-crossing permutations $NC(k)$.
\end{remark}

\subsection{Expectation of the moments of partially transposed Wishart matrices} \hfill\vspace{0.1cm}

Let $W\sim\mathcal W_{d^2,s}$ be a Wishart matrix of parameters $(d^2,s)$, meaning that $W = GG^*\in M_{d^2}(\mathbb C)$, where $G\in M_{d^2,s}(\mathbb C)$ has independent $\mathcal{N}_{\C}(0,1)$ entries. We recall that the partial transposition of such Wishart matrix $W$, viewed as an element of $M_d(\C)\otimes M_d(\C)$, is defined as
\[ W^\Gamma := \left[\mathrm{id}_d \otimes\mathrm{T}_d\right](W). \]
Note that $W$ is PSD by construction while $W^\Gamma$ is in general not PSD. Given $k\in\mathbb N$, we now want to compute the quantity $\E[\Tr((W^\Gamma)^k)]$. We recall that we define $\gamma \in S_k$ as the $k$-cycle permutation $\gamma = (k, k-1, \dots, 2, 1)$, so that its inverse is $\gamma^{-1} = (1, 2, \dots,k-1, k)$.

\begin{proposition} \label{prop:first_moment}
Let $W \sim \mathcal W_{d^2, s}$. For any $k\in\mathbb N$, the average of the trace of $(W^\Gamma)^k$ is given by
\begin{equation} \label{eq:first_moment_formula}
\E \left[\Tr\left(\left(W^\Gamma\right)^k\right)\right] = \sum_{\pi \in S_k} d^{\#(\gamma^{-1}\pi) + \#(\pi\gamma)} s^{\#(\pi)}.
\end{equation}
\end{proposition}

\begin{proof}
We write $W = GG^*$, where $G \in M_{d^2, s}(\C)$ has independent $\mathcal{N}_{\C}(0,1)$ entries. The matrix elements of $W$ are $W_{uv} = \sum_{l=1}^s G_{ul} \overbar{G}_{vl}$, , $1\leq u,v\leq d^2$, where an index $u\in\{1, \dots, d^2\}$ corresponds to a pair of indices $(i,a)\in \{1, \dots, d\}\times\{1, \dots, d\}$. The partial transpose acts on indices as $(W^\Gamma)_{(i,a), (j,b)} = W_{(i,b),(j,a)}$.

Expanding the trace, we thus get
\begin{align*}
\Tr\left(\left(W^\Gamma\right)^k\right) & = \sum_{\substack{i_1, \dots, i_k =1 \\ a_1, \dots, a_k =1}}^d W^\Gamma_{(i_1,a_1), (i_2,a_2)} \cdots W^\Gamma_{(i_k,a_k), (i_1,a_1)} \\
& = \sum_{\substack{i_1, \dots, i_k =1 \\ a_1, \dots, a_k =1}}^d W_{(i_1, a_2), (i_2, a_1)} \cdots W_{(i_k, a_1), (i_1, a_k)} \\
& =
\sum_{\substack{i_1,\ldots,i_k=1 \\ a_1,\ldots,a_k=1}}^d
\sum_{l_1,\ldots,l_k=1}^s
\prod_{t=1}^k G_{(i_t, a_{t+1}), l_t} \overbar{G}_{(i_{t+1}, a_t), l_t},
\end{align*}
where indices $t$ are understood cyclically (i.e.~$k+1 \equiv 1$). Setting $\gamma = (k, k-1, \dots, 1)\in S_k$, we can rewrite this as
\[
\Tr\left(\left(W^\Gamma\right)^k\right) =
\sum_{\substack{i_1,\ldots,i_k=1 \\ a_1,\ldots,a_k=1}}^d
\sum_{l_1,\ldots,l_k=1}^s
\prod_{t=1}^k G_{(i_t, a_{\gamma^{-1}(t)}), l_t}
\overbar{G}_{(i_{\gamma^{-1}(t)}, a_t), l_t}.
\]

We now compute the expectation using Wick's formula for complex Gaussian variables, recalled in Theorem \ref{thm:wick_complex_general}, and the fact that the entries $G_{(i,a),l}$ are independent $\mathcal{N}_\C(0,1)$ variables. We get
\begin{align*}
    \E \left[ \prod_{t=1}^k G_{(i_t, a_{\gamma^{-1}(t)}), l_t} \overbar{G}_{(i_{\gamma^{-1}(t)}, a_t), l_t} \right] & = \sum_{\pi \in S_k} \prod_{t=1}^k \E \left[ G_{(i_t, a_{\gamma^{-1}(t)}), l_t}
    \overbar{G}_{(i_{\gamma^{-1}(\pi(t))}, a_{\pi(t)}), l_{\pi(t)}} \right]  \\
    & = \sum_{\pi \in S_k} \prod_{t=1}^k \delta_{i_t, i_{\gamma^{-1}(\pi(t))}} \delta_{a_{\gamma^{-1}(t)}, a_{\pi(t)}} \delta_{l_t, l_{\pi(t)}}.
\end{align*}
Substituting this back into the expectation of the trace and summing over all indices $i_t, a_t \in \{1, \dots, d\}$ and $l_t \in \{1, \dots, s\}$, we eventually obtain
\[
\E \left[\Tr\left(\left(W^\Gamma\right)^k\right)\right]
=
\sum_{\substack{i_1,\ldots,i_k=1 \\ a_1,\ldots,a_k=1}}^d
\sum_{l_1,\ldots,l_k=1}^s
\sum_{\pi \in S_k}
\prod_{t=1}^k
\delta_{i_t, i_{\gamma^{-1}(\pi(t))}}
\delta_{a_{\gamma^{-1}(t)}, a_{\pi(t)}}
\delta_{l_t, l_{\pi(t)}}.
\]

For a fixed permutation $\pi \in S_k$, we see that the sum involves counting the number of free indices determined by the Kronecker delta constraints:
\begin{itemize}
    \item $\prod_{t=1}^k \delta_{l_t, l_{\pi(t)}} = 1$ requires $(l_t)_{t=1}^k$ to be constant on the cycles of $\pi$. Summing over $l_t \in \{1, \dots, s\}$ gives a contribution $s^{\#(\pi)}$.
    \item $\prod_{t=1}^k \delta_{i_t, i_{\gamma^{-1}(\pi(t))}} = 1$ requires $(i_t)_{t=1}^k$ to be constant on the cycles of $\gamma^{-1}\pi$. Summing over $i_t \in \{1, \dots, d\}$ gives a contribution $d^{\#(\gamma^{-1}\pi)}$.
    \item $\prod_{t=1}^k \delta_{a_{\gamma^{-1}(t)}, a_{\pi(t)}} = 1$ requires $a_{\gamma^{-1}(t)} = a_{\pi(t)}$. Letting $u = \gamma^{-1}(t)$, the condition becomes $a_u = a_{\pi(\gamma(u))} = a_{(\pi\gamma)(u)}$. This requires $(a_u)_{u=1}^k$ to be constant on the cycles of $\pi\gamma$. Summing over $a_u \in \{1, \dots, d\}$ gives a contribution $d^{\#(\pi\gamma)}$.
\end{itemize}

Combining these factors obtained for each $\pi \in S_{k}$ yields 
\begin{equation*}
\E \left[\Tr\left(\left(W^\Gamma\right)^k\right)\right] = \sum_{\pi \in S_k} d^{\#(\gamma^{-1}\pi) + \#(\pi\gamma)} s^{\#(\pi)},
\end{equation*}
which is exactly the announced formula.
\end{proof}

\begin{theorem} \label{th:first_moment_asymptotic}
Fix $\lambda>0$ and let $W \sim \mathcal W_{d^2, \lambda d^2}$. For any $k\in\mathbb N$, the asymptotic value (as $d\to\infty$) of the average of the trace of $(W^\Gamma)^k$ is given by
\begin{equation} \label{eq:first_moment_asymptotic}
\E \left[\Tr\left(\left(W^\Gamma\right)^k\right)\right] \underset{d\to\infty}{=} \left(\sum_{l=0}^{\lfloor k/2\rfloor} \binom{k}{2l}\mathrm{Cat}(l)\lambda^{k-l}\right) d^{2k+2} + O\left(d^{2k}\right).
\end{equation}
\end{theorem}

\begin{proof}
We start from the exact formula given by Proposition \ref{prop:first_moment}, applied to the case where $s=\lambda d^2$, which gives
    \[
    \E \left[\Tr\left(\left(W^\Gamma\right)^k\right)\right] = \sum_{\pi \in S_k} \lambda^{\#(\pi)}d^{\#(\gamma^{-1}\pi) + \#(\pi\gamma) + 2\#(\pi)}.
    \]
We know by inequality \eqref{eq:triangle-inequality-gamma-id} that, for all $\pi\in S_k$, $\#(\gamma^{-1}\pi) + \#(\pi) \leq k + 1$ and $\#(\pi\gamma) + \#(\pi) \leq k + 1$, with equality in the first, resp.~second, inequality if and only if $\pi$ is on the geodesic between $\gamma$, resp.~$\gamma^{-1}$, and $\mathrm{id}$. Both geodesics correspond to permutations whose associated cycle partition is non-crossing, with the ordering inside cycles given by that of either $\gamma$ or $\gamma^{-1}$. Hence, for all
$\pi\in S_k$, we have
\begin{equation} \label{eq:double-triangle-inequality'} 
\#(\gamma^{-1}\pi) + \#(\pi\gamma) + 2\#(\pi) \leq 2k+2, 
\end{equation}
with equality if and only if $\pi$ has a non-crossing cycle partition with cycles ordered both according to $\gamma$ and $\gamma^{-1}$. This is possible if and only if the cycles of $\pi$, in addition to forming a non-crossing partition, all have length $1$ or $2$.

Denoting by $NC_{\leq 2}(k)$ the set of such permutations, we thus have
    \begin{equation} \label{eq:first_moment_asymptotic'}
        \E \left[\Tr\left(\left(W^\Gamma\right)^k\right)\right] \underset{d\to\infty}{\sim} \left( \sum_{\pi \in NC_{\leq 2}(k)} \lambda^{\#(\pi)}\right) d^{2k+2}.
    \end{equation}
Now, a procedure to construct $\pi\in NC_{\leq 2}(k)$ can be described as follows:
    \begin{enumerate}
        \item Fix $l\in\{0,\ldots,\lfloor k/2\rfloor\}$.
        \item Choose the $2l$ elements of $k$ that belong to cycles of length $2$. There are $\binom{k}{2l}$ ways of doing this.
        \item Pair these $2l$ elements into a non-crossing pairing. There are $\mathrm{Cat}(l)$ ways of doing this.
    \end{enumerate}
Such $\pi$ then has $k-l$ cycles ($l$ cycles of length $2$ and $k-2l$ cycles of length $1$). We can therefore rewrite the coefficient in front of $d^{2k+2}$ in equation \eqref{eq:first_moment_asymptotic'} exactly as announced.

Let us now justify that the error term is indeed $O(d^{2k})$. If inequality \eqref{eq:double-triangle-inequality'} is strict, then either $\#(\gamma^{-1}\pi) + \#(\pi) < k+1$ or $\#(\pi\gamma) + \#(\pi) < k+1$. Supposing for instance that we are in the former case (the reasoning being exactly the same for the latter), this actually necessarily implies (as explained in Section \ref{sec:combinatorics}) that $\#(\gamma^{-1}\pi) + \#(\pi) \leq k-1$, and therefore
\[  \#(\gamma^{-1}\pi) + \#(\pi\gamma) + 2\#(\pi) \leq 2k, \]
which is precisely what we wanted to show.
\end{proof}

\subsection{Variance of the moments of partially transposed Wishart matrices} \hfill\vspace{0.1cm}

Given $k\in\mathbb N$, we now want to compute the quantity $\E [ (\Tr((W^\Gamma)^k))^2 ]$. We define two permutations in $S_{2k}$ that will play an important role later on:
\begin{itemize}
    \item $\gamma_1 = (k, k-1, \dots, 1)$ acting on $\{1, \dots, k\}$ as a $k$-cycle (and leaving $\{k+1, \dots, 2k\}$ unchanged);
    \item $\gamma_2 = (2k, 2k-1, \dots, k+1)$ acting on $\{k+1, \dots, 2k\}$ as a $k$-cycle (and leaving $\{1, \dots, k\}$ unchanged).
\end{itemize}
Let $\tilde\gamma = \gamma_1 \gamma_2$. Note that $\gamma_1$ and $\gamma_2$ commute and that $\tilde\gamma^{-1} = \gamma_1^{-1} \gamma_2^{-1}$.

\begin{proposition} \label{prop:second_moment}
Let $W \sim \mathcal W_{d^2, s}$. For any $k\in\mathbb N$, the average of the trace of $(W^\Gamma)^k$ squared is given by
\begin{equation} \label{eq:second_moment_formula}
\E \left[ \left(\Tr\left(\left(W^\Gamma\right)^k\right)\right)^2 \right] = \sum_{\pi \in S_{2k}} d^{\#(\tilde\gamma^{-1}\pi) + \#(\pi\tilde\gamma)} s^{\#(\pi)}.
\end{equation}
\end{proposition}

\begin{proof}
As in the proof of Proposition \ref{prop:first_moment}, we write $W = GG^*$, where $G\in M_{d^2, s}(\C)$ has independent $\mathcal{N}_{\C}(0,1)$ entries. First, we express the squared trace as a single sum by duplicating indices. We let the indices for the first trace factor run from $t=1$ to $t=k$ and for the second trace factor from $t=k+1$ to $t=2k$. We thus get
\begin{align*}
    \left(\Tr\left(\left(W^\Gamma\right)^k\right)\right)^2
    = & \left(
        \sum_{\substack{i_1,\ldots,i_k=1 \\ a_1,\ldots,a_k=1}}^d
        \sum_{l_1,\ldots,l_k=1}^s
        \prod_{t=1}^k
        G_{(i_t, a_{t+1}), l_t}
        \overbar{G}_{(i_{t+1}, a_t), l_t}
    \right) \times \\
& \left(
        \sum_{\substack{i_{k+1},\ldots,i_{2k}=1 \\
                        a_{k+1},\ldots,a_{2k}=1}}^d
        \sum_{l_{k+1},\ldots,l_{2k}=1}^s
        \prod_{t=k+1}^{2k}
        G_{(i_t, a_{t+1}), l_t}
        \overbar{G}_{(i_{t+1}, a_t), l_t}
    \right),
\end{align*}
where the indices $t$ in $a_{t+1}, i_{t+1}$ are interpreted cyclically within $\{1, \dots, k\}$ for the first product and within $\{k+1, \dots, 2k\}$ for the second product. Setting $\tilde\gamma = (k, k-1, \dots, 1)(2k, 2k-1, \dots, k+1)\in S_{2k}$, we can combine the two products and rewrite this as
\[
\left(\Tr\left(\left(W^\Gamma\right)^k\right)\right)^2
=
\sum_{\substack{i_1,\ldots,i_{2k}=1 \\
                a_1,\ldots,a_{2k}=1}}^d
\sum_{l_1,\ldots,l_{2k}=1}^s
\prod_{t=1}^{2k}
G_{(i_t, a_{\tilde\gamma^{-1}(t)}), l_t}
\overbar{G}_{(i_{\tilde\gamma^{-1}(t)}, a_t), l_t}.
\]

We now apply Wick's formula, recalled in Theorem \ref{thm:wick_complex_general}, to compute the expectation of this product of $4k$ complex Gaussian variables ($2k$ variables $G_{(i,a),l}$ and their $2k$ conjugates $\overbar{G}_{(i,a),l}$). Using the fact that the entries $G_{(i,a),l}$ are independent $\mathcal{N}_\C(0,1)$ variables, we get
\begin{align*}
\E \left[
    \prod_{t=1}^{2k}
    G_{(i_t, a_{\tilde\gamma^{-1}(t)}), l_t}
    \overbar{G}_{(i_{\tilde\gamma^{-1}(t)}, a_t), l_t}
\right]
&= \sum_{\pi \in S_{2k}} \prod_{t=1}^{2k}
    \E \left[
        G_{(i_t, a_{\tilde\gamma^{-1}(t)}), l_t}
        \overbar{G}_{(i_{\tilde\gamma^{-1}(\pi(t))}, a_{\pi(t)}), l_{\pi(t)}}
    \right] \\
&= \sum_{\pi \in S_{2k}} \prod_{t=1}^{2k}
    \delta_{(i_t, a_{\tilde\gamma^{-1}(t)}),
    (i_{\tilde\gamma^{-1}(\pi(t))}, a_{\pi(t)})}
    \delta_{l_t, l_{\pi(t)}} \\
&= \sum_{\pi \in S_{2k}} \prod_{t=1}^{2k}
    \delta_{i_t, i_{\tilde\gamma^{-1}(\pi(t))}}
    \delta_{a_{\tilde\gamma^{-1}(t)}, a_{\pi(t)}}
    \delta_{l_t, l_{\pi(t)}}.
\end{align*}

Substituting this back into the expectation of the squared trace, we eventually obtain
\[
\E \left[
    \left(\Tr\left(\left(W^\Gamma\right)^k\right)\right)^2
\right]
=
\sum_{\substack{i_1,\ldots,i_{2k}=1 \\ a_1,\ldots,a_{2k}=1}}^d
\sum_{l_1,\ldots,l_{2k}=1}^s
\sum_{\pi \in S_{2k}} \prod_{t=1}^{2k}
\delta_{i_t, i_{\tilde\gamma^{-1}(\pi(t))}}
\delta_{a_{\tilde\gamma^{-1}(t)}, a_{\pi(t)}}
\delta_{l_t, l_{\pi(t)}} .
\]

For a fixed permutation $\pi \in S_{2k}$, we perform the summation over the indices $i_t, a_t, l_t$. The product of Kronecker deltas is $1$ only if the indices satisfy certain equality constraints defined by the permutations. We count the number of free indices:
\begin{itemize}
    \item The constraint $\prod_{t=1}^{2k} \delta_{l_t, l_{\pi(t)}} = 1$ requires $l_t = l_{\pi(t)}$ for all $t$. Thus $(l_t)_{t=1}^{2k}$ must be constant along the cycles of $\pi$. Summing over $l_t \in \{1, \dots, s\}$ yields a factor $s^{\#(\pi)}$.
    \item The constraint
    $\prod_{t=1}^{2k} \delta_{i_t, i_{\tilde\gamma^{-1}(\pi(t))}} = 1$ requires $i_t = i_{\tilde\gamma^{-1}(\pi(t))}$ for all $t$. Thus $(i_t)_{t=1}^{2k}$ must be constant along the cycles of $\tilde\gamma^{-1}\pi$, giving a factor $d^{\#(\tilde\gamma^{-1}\pi)}$.
    \item The constraint $\prod_{t=1}^{2k} \delta_{a_{\tilde\gamma^{-1}(t)}, a_{\pi(t)}} = 1$ requires $a_{\tilde\gamma^{-1}(t)} = a_{\pi(t)}$ for all $t$. Equivalently, with $u=\tilde\gamma^{-1}(t)$, it requires $a_u=a_{\pi(\tilde\gamma(u))}=a_{(\pi\tilde\gamma)(u)}$ for all $u$. Thus $(a_u)_{u=1}^{2k}$ must be constant along the cycles of $\pi\tilde\gamma$, giving a factor $d^{\#(\pi\tilde\gamma)}$.
\end{itemize}

Combining these factors obtained for each $\pi \in S_{2k}$ yields
\begin{equation*}
\E \left[ \left(\Tr\left(\left(W^\Gamma\right)^k\right)\right)^2 \right]  = \sum_{\pi \in S_{2k}} d^{\#(\tilde\gamma^{-1}\pi) + \#(\pi\tilde\gamma)} s^{\#(\pi)},
\end{equation*}
which is exactly the announced formula.
\end{proof}

\begin{theorem} \label{th:second_moment_asymptotic}
Fix $\lambda>0$ and let $W \sim \mathcal W_{d^2, \lambda d^2}$. For any $k\in\mathbb N$, the asymptotic value (as $d\to\infty$) of the average of the trace of $(W^\Gamma)^k$ squared is given by
\begin{equation} \label{eq:second_moment_asymptotic}
\E \left[ \left(\Tr\left(\left(W^\Gamma\right)^k\right)\right)^2 \right]
\underset{d\to\infty}{=}
\left(
    \sum_{l=0}^{\lfloor k/2\rfloor}
    \binom{k}{2l}\mathrm{Cat}(l)\lambda^{k-l}
\right)^2 d^{4k+4}
+ O\left(d^{4k+2}\right).
\end{equation}
\end{theorem}

\begin{proof}
The argument is very similar to the one developed in the proof of Theorem \ref{th:first_moment_asymptotic}. We might therefore skip a few details here and reuse some of the notations introduced there. We start from the exact formula given by Proposition \ref{prop:second_moment}, applied to the case where $s=\lambda d^2$, which gives
\[
\E \left[ \left(\Tr\left(\left(W^\Gamma\right)^k\right)\right)^2 \right]
=
\sum_{\pi \in S_{2k}}
\lambda^{\#(\pi)}
d^{\#(\tilde\gamma^{-1}\pi) + \#(\pi\tilde\gamma) + 2\#(\pi)} .
\]
We know by inequality \eqref{eq:triangle-inequality} that, since $\#(\tilde\gamma)=2$, for all $\pi\in S_{2k}$, $\#(\tilde\gamma^{-1}\pi) + \#(\pi) \leq 2k+2$ and $\#(\pi\tilde\gamma) + \#(\pi) \leq 2k + 2$, with equality in the first, resp.~second, inequality if and only if $\pi$ is on the geodesic between $\tilde\gamma$, resp.~$\tilde\gamma^{-1}$, and $\mathrm{id}$. Both geodesics correspond to permutations that factorize as $\pi_1\pi_2$ with $\pi_1$, resp.~$\pi_2$, acting non-trivially only on $\{1,\ldots,k\}$, resp.~$\{k+1,\ldots,2k\}$, and whose cycle partitions are both non-crossing (with the ordering inside cycles given by that of either $\gamma_1$, resp.~$\gamma_2$, or $\gamma^{-1}_1$, resp.~$\gamma^{-1}_2$). Hence, for all $\pi\in S_{2k}$, we have
    \begin{equation} \label{eq:double-triangle-inequality}
    \#(\tilde\gamma^{-1}\pi) + \#(\pi\tilde\gamma) + 2\#(\pi) \leq 4k+4,
    \end{equation}
with equality if and only if $\pi=\pi_1\pi_2$ with $\pi_1\in NC_{\leq
2}(\{1,\ldots,k\})$ and $\pi_2\in NC_{\leq 2}(\{k+1,\ldots,2k\})$.

Identifying $NC_{\leq 2}(\{1,\ldots,k\})$ and $NC_{\leq 2}(\{k+1,\ldots,2k\})$ with $NC_{\leq 2}(k)$ and observing that for $\pi=\pi_1\pi_2$, $\#(\pi)=\#(\pi_1)+\#(\pi_2)$, we thus have
\begin{equation} \label{eq:second_moment_asymptotic'}
\E \left[ \left(\Tr\left(\left(W^\Gamma\right)^k\right)\right)^2 \right]
\underset{d\to\infty}{\sim}
\left(
    \sum_{\pi\in NC_{\leq 2}(k)} \lambda^{\#(\pi)}
\right)^2 d^{4k+4}.
\end{equation}
So we know by the proof of Theorem \ref{th:first_moment_asymptotic} that the coefficient in front of $d^{4k+4}$ in equation \eqref{eq:second_moment_asymptotic'} takes exactly the announced form.

Let us now justify that the error term is indeed $O(d^{4k+2})$. There are two ways in which inequality \eqref{eq:double-triangle-inequality} can be strict:
    \begin{enumerate}
        \item because $\pi$ does not factorize as $\pi_1\pi_2$ with $\pi_1$ permutation of $\{1,\ldots,k\}$ and $\pi_2$ permutation of $\{k+1,\ldots,2k\}$;
        \item because $\pi$ does factorize as such $\pi_1\pi_2$ but either $\pi_1\notin NC_{\leq 2}(\{1,\ldots,k\})$ or $\pi_2\notin NC_{\leq 2}(\{k+1,\ldots,2k\})$.
    \end{enumerate}
In the first case, we have that $\pi$ is neither on the geodesic between $\tilde\gamma$ and $\mathrm{id}$ nor on the geodesic between $\tilde\gamma^{-1}$ and $\mathrm{id}$. Hence, both $\#(\tilde\gamma^{-1}\pi) + \#(\pi) < 2k+2$ and $\#(\pi\tilde\gamma) + \#(\pi) < 2k+2$, which actually necessarily implies (as explained in Section \ref{sec:combinatorics}) that $\#(\tilde\gamma^{-1}\pi) + \#(\pi) \leq 2k$ and $\#(\pi\tilde\gamma) + \#(\pi) \leq 2k$, so that
    \[
    \#(\tilde\gamma^{-1}\pi) + \#(\pi\tilde\gamma) + 2\#(\pi) \leq 4k.
    \]
In the second case, we can write
    \[
    \#(\tilde\gamma^{-1}\pi) + \#(\pi\tilde\gamma) + 2\#(\pi) = \#(\gamma_1^{-1}\pi_1) + \#(\pi_1\gamma_1) + 2\#(\pi_1) + \#(\gamma_2^{-1}\pi_2) + \#(\pi_2\gamma_2) + 2\#(\pi_2).
    \]
Supposing for instance that $\pi_1\notin NC_{\leq 2}(\{1,\ldots,k\})$ (the reasoning being exactly the same for $\pi_2\notin NC_{\leq 2}(\{k+1,\ldots,2k\})$), this means that either $\#(\gamma_1^{-1}\pi_1) + \#(\pi_1) < k+1$ or $\#(\gamma_1\pi_1) + \#(\pi_1) < k+1$. And supposing for instance that we are in the former case (the reasoning being exactly the same for the latter), this actually necessarily implies (as explained in Section \ref{sec:combinatorics}) that $\#(\gamma_1^{-1}\pi_1) + \#(\pi_1) \leq k-1$. Hence,
    \[
    \#(\gamma_1^{-1}\pi_1) + \#(\pi_1\gamma_1) + 2\#(\pi_1) \leq 2k,
    \]
and therefore
    \[
    \#(\gamma_1^{-1}\pi_1) + \#(\pi_1\gamma_1) + 2\#(\pi_1) + \#(\gamma_2^{-1}\pi_2) + \#(\pi_2\gamma_2) + 2\#(\pi_2) \leq 4k+2.
    \]
We have thus shown that, if inequality \eqref{eq:double-triangle-inequality} is strict, then it has to be at most $4k+2$, as wanted.
\end{proof}

\section{Concentration of moments and Hankel determinants of partially transposed Wishart matrices}
\label{sec:concentration}

In this section, we use standard Gaussian concentration results to establish that the moments of partially transposed Wishart matrices concentrate around their expected value, and thus Hankel determinants formed from them as well.

\subsection{Concentration for polynomials in Gaussian variables} \hfill\vspace{0.1cm}

A fundamental result from Gaussian analysis, derived from the hypercontractivity property of the Ornstein--Uhlenbeck semigroup, provides an upper bound on the probability that polynomials in Gaussian variables deviate from their average (see e.g.~\cite[Section 5.2.4.3]{AubrunSzarek2017} for a full description of the proof ingredients and \cite[Corollary 5.49]{AubrunSzarek2017} for the statement of the final result).

\begin{theorem}[Deviation inequality for polynomials in Gaussian variables] \label{thm:poly_tail_bound}
Let $X = Q(Z_1, \dots, Z_n)$ be a polynomial of total degree $K$ in $n$ independent Gaussian variables with mean $0$ and variance $1$. Then, for any $t > 0$,
\[
\Prob\left( \left|X-\E(X)\right| \ge t \sqrt{\Var(X)} \right) \le \exp \left( - \inf_{q \ge 2} \left( q \log t - \frac{Kq}{2} \log(q-1) \right) \right).
\]
Optimizing $q$ leads to the simpler, more commonly used, deviation inequality: for any $t \ge (2e)^{K/2}$,
\begin{equation} \label{eq:poly_tail_bound_final}
\Prob\left(  \left|X-\E(X)\right| \ge t \sqrt{\Var(X)}\right) \le \exp\left(-\frac{Kt^{2/K}}{2e}\right).
\end{equation}
\end{theorem}

\begin{remark}
The above result is usually stated for polynomials in real Gaussian variables, but it holds for polynomials in complex Gaussian variables as well, by identifying $\C^n$ with $\R^{2n}$ (using real and imaginary parts, which are independent and each distributed according to $\mathcal{N}_\R(0, 1/2)$). The total degree $K$ then refers to the degree in the underlying $2n$ independent real variables. Now, if $Z = (X + i Y)/\sqrt{2}$ with $X,Y \sim \mathcal{N}_\R(0,1)$, then $Z, \overbar{Z}$ relate linearly to $X, Y$. So a polynomial of degree $K$ in $Z_j, \overbar{Z}_j$, $1\leq j\leq n$, is a polynomial of degree $K$ in $X_j, Y_j$, $1\leq j\leq n$.
\end{remark}

\subsection{Application to moments of partially transposed Wishart matrices} \hfill\vspace{0.1cm}

Let $W\sim\mathcal W_{d^2,s}$ be a Wishart matrix of parameters $(d^2,s)$, meaning that $W = GG^*\in M_{d^2}(\mathbb C)$, where $G\in M_{d^2,s}(\mathbb C)$ has independent $\mathcal{N}_{\C}(0,1)$ entries. For each $k\in\mathbb N$, we consider the random variable
\[ p_k(W) \coloneqq \Tr\left(\left(W^\Gamma\right)^k\right). \]
Each entry $W_{uv} = \sum_{l=1}^s G_{ul} \overbar{G}_{vl}$, $1\leq u,v\leq d^2$, of $W$ is a polynomial of degree $2$ in the entries of $G$ and $\overbar{G}$. Since $p_k(W) = \Tr((W^\Gamma)^k)$ is formed by summing products of $k$ entries of $W$ (with permuted indices due to $\Gamma$), $p_k(W)$ is a polynomial of degree $2k$ in the entries of $G$ and $\overbar{G}$.

We can now apply Theorem \ref{thm:poly_tail_bound} to $p_k(W)$. We set 
\[ E_k := \E\left(p_k(W)\right) \quad \text{and} \quad \sigma_k^2 := \Var\left(p_k(W)\right). \] 
We consider the regime where $d \to \infty$ and $s\sim \lambda d^2$ for some constant $\lambda>0$. In this regime, the asymptotic value of $E_k$ is given in Theorem \ref{th:first_moment_asymptotic}, while $\sigma_k^2 = \E(p_k(W)^2) - E_k^2$ is bounded by combining Theorems \ref{th:first_moment_asymptotic} and \ref{th:second_moment_asymptotic}. Concretely, setting
\begin{equation} \label{eq:def-M_k}
M_k(\lambda) := \sum_{l=0}^{\lfloor k/2\rfloor} \binom{k}{2l}\mathrm{Cat}(l)\lambda^{k-l},
\end{equation}
we have
\begin{equation} \label{eq:E_k-sigma_k}
\left|E_k - M_k(\lambda) d^{2k+2}\right| \leq C'(\lambda,k)d^{2k} \quad \text{and} \quad \sigma_k^2 \leq C(\lambda,k)^2 d^{4k+2}, 
\end{equation}
for some constants $C'(\lambda,k),C(\lambda,k)>0$ depending only on $\lambda$ and $k$.

By equation \eqref{eq:poly_tail_bound_final} in Theorem \ref{thm:poly_tail_bound}, applied to the polynomial $X = p_k(W)$ of degree $2k$, with $\E(X)=E_k$ and $\sqrt{\Var(X)} = \sigma_k$, we get that, for any $t \ge (2e)^{k}$,
\[ \Prob(|p_k(W) - E_k| > t \sigma_k) \le e^{-kt^{1/k}/e}. \]
In order to get an estimate on the probability of a relative deviation of $p_k(W)$ from its average, we now apply the above statement to 
\[ t = \frac{\epsilon M_k(\lambda) d^{2k+2}}{\sigma_k}, \]
for some fixed $\epsilon>0$. We observe that such $t$ satisfies 
\[ t \geq \frac{\epsilon M_k(\lambda) d}{C(\lambda,k)}, \]
and hence we indeed have $t\geq(2e)^k$ as soon as $\epsilon\geq C''(\lambda,k)/d$, where $C''(\lambda,k)=C(\lambda,k)(2e)^k/M_k(\lambda)$. We therefore get that, for any $\epsilon\geq C''(\lambda,k)/d$,
\begin{equation} \label{eq:pk_concentration}
\Prob\left(\left|p_k(W) - E_k\right| > \epsilon M_k(\lambda) d^{2k+2}\right)  \leq e^{-c(\lambda,k)\epsilon^{1/k}d^{1/k}}, 
\end{equation}
where we have set
\[ c(\lambda,k)=\frac{k}{e}\left(\frac{M_k(\lambda)}{C(\lambda,k)}\right)^{1/k}. \]

We summarize the above discussion in the following theorem, where we use the notations that have just been introduced.

\begin{theorem} \label{th:pk_concentration}
Fix $\lambda>0$ and let $W \sim \mathcal W_{d^2, \lambda d^2}$. Then, for any $k\in\mathbb N$, we have that, for any $\epsilon\geq 2C''(\lambda,k)/d$,
\[ \Prob\left(\left|p_k(W) - M_k(\lambda) d^{2k+2}\right|> \epsilon M_k(\lambda) d^{2k+2}\right)  \leq e^{-c'(\lambda,k)\epsilon^{1/k}d^{1/k}}, \]
where $c'(\lambda,k)=2^{1/k}c(\lambda,k)$.
\end{theorem}

\begin{proof}
Theorem \ref{th:pk_concentration} is a straightforward consequence of the deviation probability given in equation \eqref{eq:pk_concentration}. We just have to observe that, by the estimate on $E_k$ given in equation \eqref{eq:E_k-sigma_k}, we have that, for any $\epsilon\geq C'(\lambda,k)/(4d^2)$,
\begin{align*} 
\Prob\left(\left|p_k(W) - M_k(\lambda) d^{2k+2}\right|> \epsilon M_k(\lambda) d^{2k+2}\right)  & \leq \Prob\left(\left|p_k(W) - E_k\right| > \left(\epsilon - \frac{C'(\lambda,k)}{d^2}\right) M_k(\lambda) d^{2k+2} \right) \\
& \leq \Prob\left(\left|p_k(W) - E_k\right| > \frac{\epsilon}{2}M_k(\lambda) d^{2k+2} \right).
\end{align*}
And since for $d$ large enough $2C''(\lambda,k)/d\geq C'(\lambda,k)/(4d^2)$, we get exactly the announced result by equation \eqref{eq:pk_concentration} applied with $\epsilon/2$ playing the role of $\epsilon$.
\end{proof}

\begin{remark} \label{rem:pk_concentration}
    In Appendix \ref{app:second_approach} we actually establish an improved version of Theorem \ref{th:pk_concentration}. Concretely, we show in Theorem \ref{th:pk_concentration_improved} that, for any $\epsilon>\hat C(\lambda,k)/d^{1+1/k}$,
    \begin{equation} \label{eq:pk_concentration_improved}
    \Prob\left(\left|p_k(W) - M_k(\lambda)d^{2k+2}\right| > \epsilon d^{2k-2} \right) \leq e^{-\hat c(\lambda,k)d^{2+2/k}\epsilon^2}, 
    \end{equation}
    where $\hat C(\lambda,k),\hat c(\lambda,k)>0$ are constants depending only on $\lambda$ and $k$. This deviation inequality is stronger than the one appearing in Theorem \ref{th:pk_concentration} since the concentration bound is exponential in $d^{2+2/k}$ instead of $d^{1/k}$. This difference is particularly significant for large $k$, where the latter decreases to $1$ (and hence tends to become trivial) while the former remains lower bounded by $d^2$. The reason why we have moved the proof of this improved statement to Appendix \ref{app:second_approach} is because it is quite long and tedious. We have thus made the choice of presenting an alternative, much more straightforward, argument in the main text, even though it gives a weaker final result.
\end{remark}

\subsection{Consequence for Hankel determinants of partially transposed Wishart matrices} \label{sec:concentration-det} \hfill\vspace{0.1cm}

Let $W\sim\mathcal W_{d^2,\lambda d^2}$ and for each $k\in\mathbb N$, define its renormalized partially transposed moment of order $k$ as
\begin{equation} \label{eq:def-tilde-p_k}
\tilde{p}_k\left(W\right) := \frac{1}{d^2}p_k\left(\frac{W}{d^2}\right) = \frac{1}{d^{2k+2}}\Tr\left(\left(W^\Gamma\right)^k\right). 
\end{equation}
We have just established, in Theorem \ref{th:pk_concentration} anr Remark \ref{rem:pk_concentration} following it, that, for any $k\in\mathbb N$, $\tilde{p}_k$ concentrates around its asymptotic average $M_k(\lambda)$. Our goal is now to derive from this fact that, for any $m\in\mathbb N$, the determinant of the Hankel matrix $\tilde{B}_m$ associated to these renormalized moments $\{\tilde p_1,\ldots,\tilde p_{2m+1}\}$ concentrates around the determinant of the Hankel matrix $H_m$ associated to the asymptotic average of these moments $\{M_1(\lambda),\ldots,M_{2m+1}(\lambda)\}$.

Let us start with a simple observation. Given $k\in\mathbb N$ and $\epsilon>0$, define the event 
\[ \Omega_k^\epsilon = \{ |\tilde p_k - M_k(\lambda)| \le \epsilon M_k(\lambda) \}, \]
and denote by $\overbar\Omega_k^\epsilon$ its complementary event. By the union bound, we then have
\[ \Prob\left(\Omega_1^\epsilon\cap\cdots\cap\Omega_k^\epsilon\right) = 1 - \Prob\left(\overbar\Omega_1^\epsilon\cup\cdots\cup\overbar\Omega_k^\epsilon\right) \geq 1 - \left(\Prob\left(\overbar\Omega_1^\epsilon\right) + \cdots + \Prob\left(\overbar\Omega_k^\epsilon\right)\right). \]
And we thus get by equation \eqref{eq:pk_concentration_improved} that, for any $\epsilon$ such that $\epsilon\geq \hat C(\lambda,i)/d^{1+1/i}$ for all $1\leq i\leq k$, i.e.~for any $\epsilon\geq \hat C(\lambda,k)/d^{1+1/k}$,
\begin{equation} \label{eq:concentration-union}
    \Prob\left(\Omega_1^\epsilon\cap\ldots\cap \Omega_k^\epsilon\right) \geq 1 - \sum_{i=1}^k e^{-\hat c(\lambda,i)d^{2+2/i}\epsilon^2} \geq 1 - ke^{-\hat c(\lambda,k)d^{2+2/k}\epsilon^2}.
\end{equation}

Given $m\in\mathbb N$, let us now suppose that the event $\Omega_1^\epsilon\cap\ldots\cap\Omega_{2m+1}^\epsilon$ holds. We then have
\begin{align*}
    \det(\tilde B_m) & = \sum_{\sigma\in S_{m+1}} \varepsilon(\sigma) \prod_{i=0}^m (\tilde B_m)_{\sigma(i),i} \\
    & = \sum_{\sigma\in S_{m+1}} \varepsilon(\sigma) \prod_{i=0}^m \tilde p_{\sigma(i)+i+1} \\
    & = \sum_{\sigma\in S_{m+1}} \varepsilon(\sigma) \prod_{i=0}^m \left(1+\delta_{\sigma(i)+i+1}\right)M_{\sigma(i)+i+1}(\lambda),
\end{align*}
where by assumption $|\delta_1|,\ldots,|\delta_{2m+1}|\leq\epsilon$. We can thus write
\begin{equation} \label{eq:deviation-det}
    \det(\tilde B_m) = \det(H_m)+\Delta_m, 
\end{equation}
with
\[ |\Delta_m| \leq \epsilon\,2^{m+1}\sum_{\sigma\in S_{m+1}}\prod_{i=0}^m M_{\sigma(i)+i+1}(\lambda) = \epsilon\, 2^{m+1}\mathrm{perm}(H_m). \]
This shows that
\begin{equation} \label{eq:concentration-det-moments}
    \Prob\left(\left|\det(\tilde B_m) - \det(H_m)\right| \leq \epsilon\, 2^{m+1}\mathrm{perm}(H_m) \right) \geq \Prob\left(\Omega_1^\epsilon\cap\cdots\cap\Omega_{2m+1}^\epsilon\right)
\end{equation}

We summarize this analysis in the following theorem.

\begin{theorem} \label{th:det_concentration}
Fix $\lambda>0$ and let $W \sim \mathcal W_{d^2, \lambda d^2}$. Then, for any $m\in\mathbb N$, we have that, for any $\epsilon\geq \tilde C(\lambda,m)/d^{1+1/(2m+1)}$,
\[ \Prob\left(\left|\det(\tilde B_m) - \det(H_m)\right| \leq \epsilon \right) \geq 1 - (2m+1)e^{-\tilde c(\lambda,m)d^{2+2/(2m+1)}\epsilon^2}, \]
where we have set
\[ \tilde C(\lambda,m) = 2^{m+1}\mathrm{perm}(H_m)\hat C(\lambda,2m+1) \quad \text{and} \quad \tilde c(\lambda,m) = \frac{\hat c(\lambda,2m+1)}{\left(2^{m+1}\mathrm{perm}(H_m)\right)^2}. \]
\end{theorem}

\begin{proof}
    Theorem \ref{th:det_concentration} immediately follows from plugging into equation \eqref{eq:concentration-det-moments} the deviation probability given in equation \eqref{eq:concentration-union}, with $\epsilon/(2^{m+1}\mathrm{perm}(H_m))$ playing the role of $\epsilon$.
\end{proof}

What Theorem \ref{th:det_concentration} tells us is that, if we have $\det(H_m)>0$, resp.~$\det(H_m)<0$, then with probability going to $1$ as $d$ grows, we have as well $\det(\tilde B_m)>0$, resp.~$\det(\tilde B_m)<0$. This shows that, in order to understand the typical sign of $\det(\tilde B_m)$, it is actually enough to understand that of $\det(H_m)$, which is what we proceed to do in the next section.

\section{Threshold result for moment-based PPT criteria} \label{sec:general_threshold}

Let $W\sim\mathcal W_{d^2,s}$ be a Wishart matrix of parameters $(d^2,s)$ with $s\sim\lambda d^2$ as $d\to\infty$, and set $\rho=W/\Tr(W)$. Recall from Definition \ref{def:p_k-PPT} that, for any given $m\in\mathbb N$, $\rho$ satisfies the $p_{2m+1}$-PPT criterion if the Hankel matrix $B_m(\rho^\Gamma)$ associated with the partially transposed moments $p_k(\rho)=\Tr((\rho^\Gamma)^k)$, $1\leq k\leq 2m+1$, is PSD. As explained in Remark \ref{rem:normalization}, this is actually equivalent to the Hankel matrix $B_m((\alpha\rho)^\Gamma)$ being PSD, for any $\alpha>0$. Choosing $\alpha=\Tr(W)/d^2$, we see that we are left with checking whether $B_m((W/d^2)^\Gamma)$ is PSD, which is itself equivalent to checking whether $B_m((W/d^2)^\Gamma)/d^2$ is PSD. And we recognize that the latter matrix is nothing else than the Hankel matrix $\tilde B_m$ associated with the renormalized moments $\{\tilde p_1,\ldots,\tilde p_{2m+1}\}$, defined in equation \eqref{eq:def-tilde-p_k}. By Sylvester's criterion, the positive semi-definiteness of $\tilde B_m$ is equivalent to the simultaneous non-negativity of all its leading principal minors, i.e.~to $\det(\tilde B_0),\dots,\det(\tilde B_m)\geq 0$.

Now, we also know from Theorem \ref{th:det_concentration} that, when the local dimension $d$ is large, for each $k\in\mathbb N$, $\det(\tilde B_k)$ concentrates around the corresponding asymptotic Hankel determinant $\det(H_k)$, where
\begin{equation}\label{eq:def-H_k}
  H_k = H_k(\lambda) := \left(M_{i+j+1}(\lambda)\right)_{0\le i,j\le k},
\end{equation}
with the asymptotic average moment coefficients $M_{i+j+1}(\lambda)$ as defined in equation \eqref{eq:def-M_k}.
The goal of this section is thus to identify, for each $k\in\mathbb N$, the roots of $\lambda\mapsto\det(H_k(\lambda))$, and hence its sign, depending on the value of $\lambda$. 

\subsection{Closed-form evaluation of asymptotic Hankel determinants} \label{sec:gen:OP} \hfill\vspace{0.1cm} %via orthogonal polynomials

The evaluation of Hankel determinants of moment sequences via orthogonal polynomials is a standard technique (see e.g.~\cite{Krattenthaler99,Krattenthaler05} for extensive surveys on determinant evaluations). More recently, closed-form expressions, in terms of Chebyshev polynomials, were obtained in \cite{CK}, for Hankel determinants of linear combinations of moments arising from weighted Motzkin and Dyck paths with constant step weights. A key observation in our case is that the moment sequence $(M_k(\lambda))_{k\in\mathbb N}$ satisfies a three-term recurrence with constant coefficients. The Hankel determinant evaluation would thus follow by applying the general theory to this particular setting. For the convenience of the reader, we nevertheless give a simple self-contained proof below, following the standard orthogonal polynomial approach.

Recall the definition of the asymptotic average moment coefficients: for all $k\in\mathbb N$,
\[ M_k(\lambda) = \sum_{l=0}^{\lfloor k/2\rfloor} \binom{k}{2l}\,\mathrm{Cat}(l)\,\lambda^{k-l}, \]
where $\mathrm{Cat}(l)$ is the $l^{\text{th}}$ Catalan number. Let
\begin{equation} \label{eq:genC-def}
    C(z):=\sum_{k\in\mathbb N}M_k(\lambda)z^k
\end{equation}
be the associated generating function.

\begin{lemma}\label{lem:genC}
The generating function $C(z)$ defined in equation \eqref{eq:genC-def} satisfies the quadratic equation
\begin{equation}\label{eq:gen:Cfunc}
  (1-\lambda z)C(z) = 1 + \lambda z^2C(z)^2.
\end{equation}
\end{lemma}

\begin{proof}
For all $n\geq 2$, the coefficient of~$z^n$ on the left hand side of equation \eqref{eq:gen:Cfunc} is $M_n(\lambda)-\lambda M_{n-1}(\lambda)$, while on the right hand side it is $\lambda\sum_{p+q=n-2}M_p(\lambda) M_q(\lambda)$. The identity
\[ M_n(\lambda) - \lambda M_{n-1}(\lambda) = \lambda\!\sum_{p+q=n-2}M_p(\lambda)M_q(\lambda) \]
follows by expanding both sides using the combinatorial definition of the coefficients $M_k(\lambda)$ and the Chu--Vandermonde identity (see Appendix \ref{app:moment} for details).
\end{proof}

We now define a linear functional $\mathcal{L}:\mathbb{R}[x]\to\mathbb{R}$ by setting, for all $k\in\mathbb N$, $\mathcal{L}[x^k]:=M_k(\lambda)$. We also define monic polynomials by the three-term recurrence
\begin{equation}\label{eq:gen:3term}
  P_0=1,\qquad P_1=x-\lambda,\qquad
  P_{n+1}=(x-\lambda)\,P_n - \lambda\,P_{n-1} \text{ for all } n\ge 1.
\end{equation}
We then define the sequence of generating functions $(F_n(z))_{n\in\mathbb N}$ by, for all $n\in\mathbb N$,
\begin{equation}
    F_n(z):=\sum_{k\in\mathbb N}\mathcal{L}\left[x^kP_n\right]z^k.
\end{equation}

\begin{lemma}\label{lem:Fn}
For all $n\in\mathbb N$,
\begin{equation}\label{eq:gen:Fn}
  F_n(z) = \lambda^nz^nC(z)^{n+1}.
\end{equation}
\end{lemma}

\begin{proof}
We prove the statement by induction. For $n=0$, we have $P_0=1$ so
$\mathcal{L}[x^kP_0]=M_{k}(\lambda)$, and hence
\[ F_0(z)=C(z). \]
For $n=1$, we have $P_1=x-\lambda$ so $\mathcal{L}[x^kP_1]=M_{k+1}(\lambda)-\lambda M_k(\lambda)$, and hence
\[ F_1(z) = \frac{1}{z}\left(C(z)-1\right)-\lambda C(z) = \frac{1}{z}\left((1-\lambda z)C(z) - 1\right) = \lambda z C(z)^2, \]
where the last equality is by Lemma \ref{lem:genC}. The initialization step is thus clear. We now turn to the induction step. The recurrence~\eqref{eq:gen:3term} gives $zF_{n+1}(z)=(1-\lambda z)F_n(z)-\lambda zF_{n-1}(z)$ for $n\ge1$. Substituting the induction hypothesis and applying Lemma \ref{lem:genC} yields $F_{n+1}(z)=\lambda^{n+1}z^{n+1}C(z)^{n+2}$, which concludes the proof.
\end{proof}

\begin{lemma}\label{lem:ortho}
For all $n\in\N$, $\mathcal{L}[x^kP_n]=0$ for all $k<n$ and $\mathcal{L}[x^nP_n]=\lambda^n$.
\end{lemma}

\begin{proof}
We know from Lemma \ref{lem:Fn} that $F_n(z)=\lambda^n z^n C(z)^{n+1}$. Hence, we first observe that $F_n(z)$ has no powers of $z$ of order smaller than $n$, which means that, for all $k<n$, the coefficient in front of $z^k$, i.e.~$\mathcal{L}[x^kP_n]$, is equal to $0$. Moreover, the coefficient in front of $z^n$, i.e.~$\mathcal{L}[x^nP_n]$, is equal to $\lambda^nC(0)^{n+1}=\lambda^n$ (because $C(0)=M_0(\lambda)=1$).
\end{proof}

As an immediate consequence of Lemma \ref{lem:ortho}, we have that, for all $k,n\in\N$,
\begin{equation} \label{eq:ortho}
    \mathcal L\left[P_kP_n\right] = \begin{cases} 0 \quad \text{if } k\neq n \\ \lambda^n \quad \text{if } k=n \end{cases}.
\end{equation}
This is because, for each $k\in\mathbb N$, $P_k$ is a monic polynomial of degree $k$.

With the above observation in mind, we now define the Gram matrix $G_m$ by, for all $0\le k,l\le m$,
\[ (G_m)_{kl}:=\mathcal{L}[xP_kP_l] . \]

\begin{lemma}\label{lem:det_eq}
For all $m\in\mathbb N$, $\det(H_m) = \det(G_m)$.
\end{lemma}

\begin{proof}
For each $0\leq k\leq m$, we write $P_k(x)=\sum_{i=0}^k S_{ki}x^i$. Since $P_k$ is a monic polynomial, we have $S_{kk}=1$. We then define the lower-triangular matrix $S=(S_{ki})_{0\leq k,i\leq m}$, which is thus such that $\det(S)=1$. Now, it is easy to check that $G_m=SH_mS^T$. Indeed, for any $0\leq k,l\leq m$, we have
\[ (SH_mS^T)_{kl} = \sum_{i=0}^k\sum_{j=0}^l S_{ki}S_{lj}(H_m)_{ij} = \sum_{i=0}^k\sum_{j=0}^l S_{ki}S_{lj}\mathcal L[x^{i+j+1}] =  \mathcal L[xP_kP_l] = (G_m)_{kl}, \]
where the third equality is by linearity of $\mathcal L$.
\end{proof}

In what follows, we denote by $(U_n)_{n\in\mathbb N}$ the sequence of Chebyshev
polynomials of the second type. We recall that they are defined by, for all $n\in\mathbb
N$ and $-\pi/2\leq\theta\leq\pi/2$,
\begin{equation} \label{eq:def-Chebyshev}
    U_n\left(\cos(\theta)\right)\sin(\theta) = \sin\left((n+1)\theta\right).
\end{equation}

\begin{theorem}\label{thm:gen:det}
For all $m\in\mathbb N$ and $\lambda>0$, we have
\[ \det \left( H_m(\lambda) \right) = \lambda^{(m+1)^2/2} U_{m+1}\left(\frac{\sqrt\lambda}{2}\right), \]
where $U_{m+1}$ is as defined in equation \eqref{eq:def-Chebyshev}.
\end{theorem}

\begin{proof}
It is easy to see that $G_m$ is tridiagonal, with non-zero entries given by
\begin{align*}
    & (G_m)_{k,k}=\lambda^{k+1} \quad \text{for all } 0\leq k\leq m,\\
    & (G_m)_{k,k-1}=\lambda^k \quad \text{for all } 1\leq k\leq m,\\
    & (G_m)_{k,k+1}=\lambda^{k+1} \quad \text{for all } 0\leq k\leq m-1.
\end{align*}
Indeed, recurrence \eqref{eq:gen:3term} tells us that $xP_k=P_{k+1}+\lambda P_k+\lambda P_{k-1}$, so the result follows from equation~\eqref{eq:ortho}. Factoring $\lambda^k$ from the $k^{\text{th}}$ row of $G_m$, for each $0\leq k\leq m$, extracts a
multiplicative factor $\prod_{k=0}^m \lambda^k=\lambda^{m(m+1)/2}$ and leaves a reduced matrix~$G'_m$ with entries $(G'_m)_{k,k-1}=1$, $(G'_m)_{k,k}=\lambda$, and $(G'_m)_{k,k+1}=\lambda$ whenever the indices are in range.

Setting $d_n:=\det (G'_{n-1})$, cofactor expansion along the first row gives
\[  d_0=1, \qquad d_1=\lambda, \qquad d_{n+1}=\lambda d_n-\lambda d_{n-1} \text{ for all } n\ge 1. \]
Defining $f_n=\lambda^{-n/2}d_n$ and $t=\sqrt\lambda/2$, we get $f_{n+1}=2tf_n-f_{n-1}$ with $f_0=1$ and $f_1=2t$, which is the defining recurrence of the Chebyshev polynomials of the second kind $U_n(t)$. This implies that $d_n=\lambda^{n/2}U_n(\sqrt\lambda/2)$.

Combining this result with Lemma \ref{lem:det_eq}, we eventually obtain
\[ \det(H_m) = \det(G_m) = \lambda^{m(m+1)/2}d_{m+1} = \lambda^{(m+1)^2/2} U_{m+1}\left(\frac{\sqrt\lambda}{2}\right), \]
as announced.
\end{proof}

\subsection{Consequence for the $p_{2m+1}$-PPT criterion threshold} \hfill\vspace{0.1cm}

\begin{theorem} \label{thm:gen-threshold}
For all $m\in\mathbb N$, setting
    \begin{equation} \label{eq:lambda-m}
        \lambda_m := 4\cos^2\left(\frac{\pi}{m+2}\right),
    \end{equation}
we have that, if $\lambda>\lambda_m$, then $\det(H_\ell(\lambda))>0$ for all $0\leq \ell\leq m$, while if $\lambda<\lambda_m$, then $\det(H_\ell(\lambda))<0$ for some $0\leq \ell\leq m$.
\end{theorem}

\begin{proof}
If $\lambda\geq 4$, then $\sqrt{\lambda}/2\geq1$ and the recurrence for the Chebyshev polynomials gives $U_n(\sqrt{\lambda}/2)>0$ for every $n\geq0$. Theorem~\ref{thm:gen:det} then implies $\det(H_\ell(\lambda))>0$ for every $0\leq \ell\leq m$. We may therefore assume that $0<\lambda<4$.

We set $\theta:=\arccos(\sqrt{\lambda}/2)$, so that $0<\theta<\pi/2$. From Theorem \ref{thm:gen:det} and the identity \eqref{eq:def-Chebyshev}, we have that, for all $0\leq \ell\leq m$,
\begin{equation}\label{eq:sign_sin}
  \operatorname{sign}\left(\det (H_\ell(\lambda))\right)
  =\operatorname{sign}\left(U_{\ell+1}\left(\frac{\sqrt{\lambda}}{2}\right)\right)
  =\operatorname{sign}\left(\sin((\ell+2)\theta)\right),
\end{equation}
because $0<\theta<\pi/2$ implies $\sin(\theta)>0$. Define next the increasing sequence $(\lambda_\ell)_{0\leq \ell\leq m}$ by, for all $0\leq \ell\leq m$,
\[ \lambda_\ell := 4\cos^2\left(\frac{\pi}{\ell+2}\right). \]
$\lambda_\ell$ is such that $\theta_\ell=\arccos(\sqrt{\lambda_\ell}/2)=\pi/(\ell{+}2)$, and hence
it is the largest zero of $\det(H_\ell(\lambda))$.

Let us first show that $\lambda>\lambda_m$ implies that $\det(H_\ell(\lambda))>0$ for all $0\leq \ell\leq m$. If $\lambda>\lambda_m$ then $\theta<\pi/(m{+}2)$. So for every $0\leq\ell\leq m$, we have $(\ell{+}2)\theta\leq(m{+}2)\theta<\pi$, hence $\sin((\ell{+}2)\theta)>0$ and $\det (H_\ell(\lambda))>0$ by equation~\eqref{eq:sign_sin}.

Let us now show that $\lambda<\lambda_m$ implies that $\det(H_\ell(\lambda))<0$ for some $0\leq \ell\leq m$. Because the sequence $(\lambda_\ell)_{0\leq \ell\leq m}$ is increasing, there
exists a unique index $1\leq \ell\leq m$ such that
\[ \lambda_{\ell-1}\le\lambda<\lambda_\ell. \]
The corresponding angle~$\theta$ satisfies
\[ \frac{\pi}{\ell+2}<\theta\leq\frac{\pi}{\ell+1}. \]
So we have on the one hand $(\ell{+}2)\theta>\pi$. On the other hand
\[ (\ell{+}2)\theta \leq \frac{(\ell{+}2)\pi}{\ell{+}1} = \pi+\frac{\pi}{\ell{+}1} \leq \frac{3\pi}{2} < 2\pi. \]
Hence $\pi<(\ell{+}2)\theta<2\pi$, so $\sin((\ell{+}2)\theta)<0$ and $\det(H_\ell(\lambda))<0$ by equation \eqref{eq:sign_sin}.
\end{proof}

\begin{corollary}\label{cor:gen-threshold}
Let $d\in\mathbb N$ and let $\rho$ be a random state on $\C^d\otimes\C^d$, induced by some environment $\C^s$ with $s\sim\lambda d^2$ as $d\to\infty$. For all $m\in\mathbb N$, defining $\lambda_m$ as in equation \eqref{eq:lambda-m}, we have
\begin{itemize}
    \item[(1)] if $\lambda < \lambda_m$, then
    \[ \Prob(\rho \text{ violates the } p_{2m+1}\text{-PPT criterion}) \geq 1-(m+1)^2e^{-c(\lambda,m)d^{2+2/(2m+1)}}; \]
    \item[(2)] if $\lambda > \lambda_m$, then
    \[ \Prob(\rho \text{ satisfies the } p_{2m+1}\text{-PPT criterion}) \geq 1-(m+1)^2e^{-c(\lambda,m)d^{2+2/(2m+1)}}, \]
\end{itemize}
where $c(\lambda,m)>0$ is a constant depending only on $\lambda$ and $m$.
\end{corollary}

What Corollary \ref{cor:gen-threshold} tells us is that the threshold environment dimension at which a random induced state on $\C^d\otimes\C^d$ switches from typically violating to typically satisfying the $p_{2m+1}$-PPT criterion is $s=\lambda_m d^2$ with $\lambda_m=4\cos^2(\pi/(m+2))$. In particular, we see that $\lambda_m\to 4$ as $m\to\infty$. This is consistent with the fact that the hierarchy of $p_{2m+1}$-PPT criteria converges to the PPT criterion as $m\to\infty$ and that the corresponding threshold for the PPT criterion is $s=4d^2$. We should however insist on the fact that doing such asymptotic extension of our results is not mathematically legit. Indeed, those are technically valid only in the regime where $m$ is fixed while $d$ grows, not in the regime where $m$ would grow with $d$ as well (e.g.~as $m\sim d^2$, which is the scaling at which we know that $p_{2m+1}$-PPT is equivalent to PPT). It remains nonetheless interesting to note that this non-rigorous extrapolation actually gives a correct result.

\begin{proof}
By the concentration result summarized in Theorem \ref{th:det_concentration} we know that, for any $0\leq \ell\leq m$, if $\det(H_\ell)>0$, then
\[ \Prob\left(\det(B_\ell)>0\right) \geq 1-(2\ell+1)e^{-c(\lambda,\ell)d^{2+2/(2\ell+1)}}, \]
while if $\det(H_\ell)<0$, then
\[ \Prob\left(\det(B_\ell)<0\right) \geq 1-(2\ell+1)e^{-c(\lambda,\ell)d^{2+2/(2\ell+1)}}. \]
By the union bound, we thus get that, if $\det(H_\ell)>0$ for all $0\leq\ell\leq m$, then
\[ \Prob\left(\forall\,0\leq\ell\leq m,\,\det(B_\ell)>0\right) \geq 1-(m+1)^2e^{-c(\lambda,m)d^{2+2/(2m+1)}}, \]
while if $\det(H_\ell)<0$ for some $0\leq\ell\leq m$, then
\[ \Prob\left(\exists\,0\leq\ell\leq m: \det(B_\ell)>0\right) \geq 1-(m+1)^2e^{-c(\lambda,m)d^{2+2/(2m+1)}}. \]
Now, by Sylvester's criterion, the fact that $\det(B_\ell)>0$ for all $0\leq\ell\leq m$ is equivalent to $B_m>0$, and thus implies $B_m\geq 0$, while the fact that $\det(B_\ell)<0$ for some $0\leq\ell\leq m$ is equivalent to $B_m\ngeq 0$. The conclusion then straightforwardly follows from Theorem \ref{thm:gen-threshold}, which tells us that, if $\lambda>\lambda_m$, then $\det(H_\ell(\lambda))>0$ for all $0\leq \ell\leq m$, while if $\lambda<\lambda_m$, then $\det(H_\ell(\lambda))<0$ for some $0\leq \ell\leq m$.
\end{proof}

\section{Extension to the unbalanced regime}
\label{sec:unbalanced}

Up to now we have only focused on the balanced regime, where the dimensions $d_1$ and $d_2$ of the two subsystems $\mathcal H_1$ and $\mathcal H_2$ are the same. The results we obtained could be easily generalized to a slightly unbalanced regime, where $d_2\sim\mu d_1$ for some fixed $\mu>0$ as $d_1,d_2\to\infty$. On the other hand, different arguments are needed in order to understand a very unbalanced regime, where say $d_1$ is fixed and only $d_2\to\infty$. In this case, it is known that the threshold for satisfying or violating the PPT criterion occurs for an environment dimension $s=2(d_1+\sqrt{d_1^2-1})d_2$ \cite{BN13}. Let us now sketch how things would go for moment-based PPT criteria. 

We start from the exact formula for the expectation of $\Tr((W^\Gamma)^k)$ established in Proposition \ref{prop:first_moment}, which we can easily generalize to the case where $d_1\neq d_2$. This gives
$$ \E \left[\Tr\left(\left(W^\Gamma\right)^k\right)\right] = \sum_{\pi \in S_k} d_1^{\#(\gamma^{-1}\pi)} d_2^{\#(\pi\gamma)} s^{\#(\pi)}. $$
As we have just recalled, we know that, for the PPT criterion, the relevant scaling for $s$ in this unbalanced regime is when $s$ is of order $d_2$, with a $d_1$-dependent coefficient that is of order $d_1$ when $d_1$ is large. Hence, inserting (as before) $s=\lambda d_1d_2$ in the above expression, we get
$$ \E \left[\Tr\left(\left(W^\Gamma\right)^k\right)\right] = \sum_{\pi \in S_k} \lambda^{\#(\pi)}d_1^{\#(\gamma^{-1}\pi)+\#(\pi)} d_2^{\#(\pi\gamma)+\#(\pi)}. $$
Now, we know that, for all $\pi\in S_k$, $\#(\pi\gamma)+\#(\pi)\leq k+1$, with equality if and only if $\pi$ has a non-crossing cycle partition with the ordering inside cycles given by $\gamma^{-1}$. Denoting by $NC_\leftarrow(k)$ the set of such permutations, we thus have
\begin{equation} \label{eq:average-unbalanced} 
\E \left[\Tr\left(\left(W^\Gamma\right)^k\right)\right] \underset{d_2\to\infty}{\sim} \left(\sum_{\pi \in NC_\leftarrow(k)} \lambda^{\#(\pi)}d_1^{\#(\gamma^{-1}\pi)+\#(\pi)} \right) d_2^{k+1} = M_k(\lambda,d_1)d_2^{k+1}, 
\end{equation}
where we have set
\begin{equation}
    M_k(\lambda,d_1) := \sum_{\pi \in NC_\leftarrow(k)} \lambda^{\#(\pi)}d_1^{\#(\gamma^{-1}\pi)+\#(\pi)}.
\end{equation}

Similarly, we can adapt the exact formula for the expectation of $(\Tr((W^\Gamma)^k))^2$ established in Proposition \ref{prop:second_moment}, which gives, for $s=\lambda d_1d_2$,
$$ \E \left[\left(\Tr\left(\left(W^\Gamma\right)^k\right)\right)^2\right] = \sum_{\pi \in S_{2k}} \lambda^{\#(\pi)} d_1^{\#(\tilde\gamma^{-1}\pi)+\#(\pi)} d_2^{\#(\pi\tilde\gamma)+\#(\pi)}. $$
And by arguments completely analogous to those used in the proof of Theorem \ref{th:second_moment_asymptotic}, we get the asymptotic value, as $d_2\to\infty$,
\begin{equation} \label{eq:variance-unbalanced} 
\E \left[\left(\Tr\left(\left(W^\Gamma\right)^k\right)\right)^2\right] \underset{d_2\to\infty}{=} M_k(\lambda,d_1)^2d_2^{2k+2} +O\left(d_2^{2k}\right). 
\end{equation}
This shows that the random variable $X=\Tr((W^\Gamma)^k)$ is such that $\E(X)$ is of order $d_2^{k+1}$ while $\sqrt{\mathrm{Var}(X)}$ is of order at most $d_2^k$, which is what allows to prove concentration of $X$ around its average, exactly as it was done in the balanced regime.

 However, it is not clear how to obtain a closed-form expression for $M_k(\lambda,d_1)$ for any $d_1\in\mathbb N$. We therefore focus here on understanding its behavior when $d_1$ is large. As we have already argued in the proof of Theorem \ref{th:first_moment_asymptotic}, we have that, for any $\pi\in NC_\leftarrow(k)$, $\#(\gamma^{-1}\pi)+\#(\pi)\leq k+1$, with equality if and only if the ordering inside cycles of $\pi$ is also given by $\gamma$, i.e.~if and only if $\pi\in NC_{\leq 2}(k)$. And as explained in Section \ref{sec:combinatorics}, if $\#(\gamma^{-1}\pi)+\#(\pi)<k+1$, then we actually necessarily have $\#(\gamma^{-1}\pi)+\#(\pi)\leq k-1$. We thus get
 \begin{equation} \label{eq:M_k-asymptotics}
     M_k(\lambda,d_1) \underset{d_1\to\infty}{=} \left(M_k(\lambda) + O\left(\frac{1}{d_1^2}\right)\right)d_1^{k+1},
 \end{equation}  
 where $M_k(\lambda)$ is the asymptotic average moment coefficient obtained in the balanced regime, and defined in equation \eqref{eq:def-M_k}.

 Following the exact same steps as in the balanced regime, we now define the asymptotic (as $d_2\to\infty$) rescaled (by $1/d_1^{k+1}$) Hankel matrix
 \[ H_k(\lambda,d_1) := \left(\frac{M_{i+j+1}(\lambda,d_1)}{d_1^{k+1}}\right)_{0\leq i,j\leq k}. \]
 As a straightforward consequence of equation \eqref{eq:M_k-asymptotics}, we have that
 \[ \det\left(H_k(\lambda,d_1)\right) \underset{d_1\to\infty}{=} \det\left(H_k(\lambda)\right) + O\left(\frac{1}{d_1^2}\right), \]
 where $H_k(\lambda)$ is the asymptotic Hankel matrix obtained in the balanced regime, and defined in equation \eqref{eq:def-H_k}.

 Now, we have established in Section \ref{sec:general_threshold} that the largest zero of $\det(H_k(\lambda))$ is simple and takes value $\lambda_k=4\cos^2(\pi/(k+2))$. By the inverse function theorem, we thus have that the largest zero of $\det(H_k(\lambda,d_1))$ satisfies
 \[ \lambda_k(d_1) \underset{d_1\to\infty}{=} 4\cos^2\left(\frac{\pi}{k+2}\right) +  O\left(\frac{1}{d_1^2}\right). \]
 
 Following the exact same argumentation as in the proofs of Theorem \ref{thm:gen-threshold} and Corollary \ref{cor:gen-threshold} we eventually obtain the following threshold statement in the unbalanced regime: Let $d_1,d_2\in\mathbb N$ and let $\rho$ be a random state on $\mathbb C^{d_1}\otimes\mathbb C^{d_2}$, induced by some environment $\mathbb C^s$ with $s\sim\lambda d_1d_2$ as $d_2\to\infty$. For all $m\in\mathbb N$, there exists $\lambda_m(d_1)$, satisfying
 \[ \lambda_m(d_1) \underset{d_1\to\infty}{=} 4\cos^2\left(\frac{\pi}{m+2}\right) +  O\left(\frac{1}{d_1^2}\right), \]
such that
\begin{itemize}
    \item[(1)] if $\lambda < \lambda_m(d_1)$, then
    \[
    \Prob(\rho \text{ violates the } p_{2m+1}\text{-PPT criterion}) \underset{d_2\to\infty}{\longrightarrow} 1;
    \]
    \item[(2)] if $\lambda > \lambda_m(d_1)$, then
    \[
    \Prob(\rho \text{ satisfies the } p_{2m+1}\text{-PPT criterion}) \underset{d_2\to\infty}{\longrightarrow} 1.
    \]
\end{itemize}

\begin{remark}
    Finding the exact expression for the threshold $\lambda_m(d_1)$, valid for all $d_1$, seems challenging. Nevertheless, pushing computations of $M_k(\lambda,d_1)$ and $H_k(\lambda,d_1)$ to the next order in $d_1$ is cumbersome but doable. This yields the more precise asymptotic expansion 
    \[ \lambda_m(d_1) \underset{d_1\to\infty}{=} 4\cos^2\left(\frac{\pi}{m+2}\right) - \frac{1}{d_1^2} +  O\left(\frac{1}{d_1^4}\right). \]
    It is interesting to observe that the first correction term $-1/d_1^2$ to the leading term $\lambda_m$ is actually independent of $m$. And it is in fact the same as in the threshold for the PPT criterion, which, as we have recalled, is $2(1+\sqrt{d_1^2-1}/d_1)=4-1/d_1^2+O(1/d_1^4)$. 
\end{remark}

\section{Summary and perspectives}
\label{sec:conclusion}

In this work, we have analyzed the performance of the hierarchy of so-called $p_{2m+1}$-PPT criteria, which are moment-based relaxations of the PPT criterion, on random high-dimensional bipartite states. We have established that, for any fixed level $m$ in the hierarchy, a random state on $\mathbb C^d\otimes\mathbb C^d$ that is induced by an environment $\mathbb C^s$ typically violates, resp.~satisfies, the $p_{2m+1}$-PPT criterion if $s<\lambda_m d^2$, resp.~$s<\lambda_m d^2$, where  
\[ \lambda_m = 4\cos^2\left(\frac{\pi}{m+2}\right). \]
This sequence of threshold values interpolates between the threshold value $\lambda_1=1$ for the $p_3$-PPT criterion and the threshold value $\lambda_\infty=4$ for the PPT criterion. One important point to note is that even the first level in this hierarchy of moment-based relaxations of PPT does not perform much worse as PPT on typical instances, in the sense that detection occurs up to $s$ of order $d^2$, as for PPT. As a comparison, the precise generic detection threshold of the $p_3$-PPT criterion, namely $s=d^2$, is already higher than that of the realignment criterion, established in \cite{AN12}, namely $s=(8/(3\pi))^2d^2$.

Several questions remain open. One natural direction would be to design other experimentally accessible relaxations of widely-used entanglement criteria, such as PPT or realignment. Indeed, both these criteria require having access to the full spectrum of a certain rewiring of a bipartite state, which is unfeasible in practice. One could then again ask what is the typical performance of such relaxed criteria when applied to random bipartite states. For instance, some first proposals for moment-based versions of the realignment criterion have appeared in \cite{Zhang2022}, but a more systematic approach to this question is still lacking. On the other hand, an even further relaxation of the $p_{2m+1}$-PPT hierarchy was very recently proposed in \cite{ME26} (where each criterion consists in an inequality involving only three partially transposed moments, as the $p_{3}$-PPT criterion), and would deserve a more in-depth analysis.

Another aspect that one could explore is whether it would be possible to estimate the size of the set of states that satisfy the $p_{2m+1}$-PPT criterion. Standard size parameters include the volume-radius (which is more geometrical) or the mean width (which is more probabilistic). Both have been estimated for the set of states satisfying the full PPT criterion \cite{AS05}. The analysis relied quite crucially on the fact that the set of PPT states can be very simply described from a geometric point of view, as the intersection between the set of all states and its image under partial transposition. In contrast, we do not have any nice geometric characterization (or even just approximation) of the set of $p_{2m+1}$-PPT states, which makes the question of computing its size quite challenging.

On a final, more prospective, note, let us mention potential investigations that could be carried on related to undistillability. We know that the set of PPT states is included inside the set of undistillable states but we do not know whether this inclusion is strict. A question we could naturally ask is whether there exists $m<d^2$ such that the set of $p_{2m+1}$-PPT states is included inside the set of undistillable states. A positive answer to this question would imply that PPT states are strictly included in undistillable states (since PPT states are strictly included in $p_{2m+1}$-PPT states for $m<d^2$). So such statement is unlikely to be (easily) provable. Conversely, on could wonder whether there exists $m\geq 1$ such that the set of undistillable states is included inside the set of $p_{2m+1}$-PPT states. A positive answer to this question would imply an approximation from the outside of undistillable states (in addition to the already known one from the inside). More generally, it would be interesting to investigate whether any fixed level in the hierarchy of $p_{2m+1}$-PPT criteria has an operational meaning, either related or not to distillation.

\subsection*{Acknowledgments} Kieran McShane is supported by the France 2030 program QuanTEdu-France (n\textdegree ANR-22-CMAS-0001) and the ANR PRC project TAGADA (n\textdegree ANR-25-CE40-5672). Cécilia Lancien is supported by the ANR JCJC project PraQPV (n\textdegree ANR-24-CE47-3023), the ANR JCJC project RTFPQuEnt (n\textdegree ANR-25-CE40-5465) and the ANR PRC project TAGADA (n\textdegree ANR-25-CE40-5672).

\appendix

\section{Alternative approach to concentration through locally Lipschitz functions}
\label{app:second_approach}

This appendix provides an alternative way of proving concentration of moments of partially transposed random states. It relies on the observation that such moments are Lipschitz functions in random matrices that are uniformly distributed on the Hilbert--Schmidt unit sphere, with a Lipschitz constant that is `small' when restricting to a subset that has `large' measure. 

\subsection{Lévy's lemma for locally Lipschitz functions} \hfill\vspace{0.1cm}

We start by recalling standard concentration results for Lipschitz functions on the Euclidean unit sphere, which our upcoming analysis will heavily rely on. The reader is for instance referred to \cite[Section 5.2]{AubrunSzarek2017} for much more details on this topic. The first such result is the well-known Lévy's lemma. We state here both its original form, in terms of median, and the one that can be easily derived from it, in terms of mean, which will be most useful to us (see e.g.~\cite[Corollary 5.17 and Proposition 5.20]{AubrunSzarek2017}).

\begin{theorem}[Lévy's lemma] \label{th:Levy} 
    Let $f:S^{D-1}\to\mathbb R$ be an $L$-Lipschitz function (with respect to the Euclidean norm). Then, for every $\epsilon>0$,
    \[ \Prob\left(f\gtrless M_f\pm\epsilon\right) \leq \frac{1}{2}\exp\left(-\frac{D\epsilon^2}{2L^2}\right) \quad \text{and} \quad \Prob\left(f\gtrless\E(f)\pm\epsilon\right) \leq \exp\left(-\frac{D\epsilon^2}{2L^2}\right). \]
\end{theorem}

We will also need a refined version of Lévy's Lemma, which is suited to the case where the considered function has a global Lipschitz constant that is too large for practical purposes, while its restriction to a subset of large measure has a much smaller Lipschitz constant (see e.g.~\cite[Corollary 5.35]{AubrunSzarek2017}).

\begin{theorem}[Lévy's lemma, restricted version] \label{th:Levy-local} 
    Let $\Omega\subset S^{D-1}$ be a subset of measure at least $3/4$ and let $f:S^{D-1}\to\mathbb R$ be a function such that the restriction of $f$ to $\Omega$ is $L$-Lipschitz (with respect to the Euclidean norm). Then, for every $\epsilon>0$,
    \[ \Prob\left(\left|f-M_f\right|>\epsilon\right) \leq \Prob\left(S^{D-1}\setminus\Omega\right) + 2\exp\left(-\frac{D\epsilon^2}{4L^2}\right). \]
\end{theorem}

In fact, once again, what we will actually require is a version of the above result stated in terms of mean rather median. We derive such statement below.

\begin{corollary} \label{cor:mean-median}
    Let $\Omega\subset S^{D-1}$ be a subset of measure at least $3/4$ and let $f:S^{D-1}\to\mathbb R$ be a function such that the restriction of $f$ to $\Omega$ is $L$-Lipschitz (with respect to the Euclidean norm). Suppose additionally that $|f|$ is upper bounded by $R$ on $S^{D-1}$. Then, for every $\epsilon>\varepsilon(D,\Omega,L,R)$,
    \[ \Prob\left(\left|f-\E(f)\right|>\epsilon\right) \leq \Prob\left(S^{D-1}\setminus\Omega\right) + 2\exp\left(-\frac{D\epsilon^2}{16L^2}\right), \]
    where we have set
    \[ \varepsilon(D,\Omega,L,R) = 4\left(R\Prob\left(S^{D-1}\setminus\Omega\right) + \frac{\sqrt{\pi}L}{\sqrt{D}}\right). \]
\end{corollary}

\begin{proof}
    We start by observing that, if $|f|$ is upper bounded by $R$ on $S^{D-1}$, then Theorem \ref{th:Levy-local} implies that
    \begin{equation} \label{eq:mean-median}
        \left|\mathbb E(f)-M_f\right| \leq 2\left(R\Prob\left(S^{D-1}\setminus\Omega\right) + \frac{\sqrt{\pi}L}{\sqrt{D}}\right).
    \end{equation}  
    Indeed, by Jensen's inequality, we first have
    \[ \left|\mathbb E(f)-M_f\right| = \left|\mathbb E(f-M_f)\right| \leq \mathbb E\left(\left|f-M_f\right|\right). \]
    Now, the fact that $|f|\leq R$ clearly implies that $|M_f|\leq R$, and therefore by the triangle inequality $|f-M_f|\leq|f|+|M_f|\leq 2R$. Hence,
    \begin{align*}
        \mathbb E\left(\left|f-M_f\right|\right) & = \int_0^{2R} \Prob\left(\left|f-M_f\right|>t\right)dt \\
        & \leq \int_0^{2R} \left(\Prob\left(S^{D-1}\setminus\Omega\right) + 2\exp\left(-\frac{Dt^2}{4L^2}\right)\right)dt \\
        & \leq 2R\Prob\left(S^{D-1}\setminus\Omega\right) + \frac{2\sqrt{\pi}L}{\sqrt{D}},
    \end{align*}
    where the first inequality is by Theorem \ref{th:Levy-local} and the second inequality is because, for all $\sigma>0$, $2\int_0^{2R}e^{-t^2/2\sigma^2}dt \leq 2\int_0^\infty e^{-t^2/2\sigma^2}dt = \sqrt{2\pi\sigma^2}$. So inequality \eqref{eq:mean-median} does hold.

    As a consequence, we have that, setting
    \[ \varepsilon(D,\Omega,L,R) = 4\left(R\Prob\left(S^{D-1}\setminus\Omega\right) + \frac{\sqrt{\pi}L}{\sqrt{D}}\right), \]
    for any $\epsilon>\varepsilon(D,\Omega,L,R)$,
    \begin{align*}
        \Prob\left(\left|f-\E(f)\right|>\epsilon\right) & \leq \Prob\left(\left|f-M_f\right|>\epsilon-\frac{\varepsilon(D,\Omega,L,R)}{2}\right) \\
        & \leq \Prob\left(\left|f-M_f\right|>\frac{\epsilon}{2}\right) \\
        & \leq \Prob\left(S^{D-1}\setminus\Omega\right) + 2\exp\left(-\frac{D\epsilon^2}{16L^2}\right),
    \end{align*} 
    where the first inequality is by inequality \eqref{eq:mean-median} and the last inequality is by Theorem \ref{th:Levy-local}.
\end{proof}

In what follows, we will work in the space of $d^2\times s$ complex matrices, which we can identify with $\R^{2d^2s}$, and whose Euclidean norm is the Hilbert--Schmidt norm (or Schatten $2$-norm), defined by, for all $X\in M_{d^2,s}(\C)$,
\[ \|X\|_2 := \left(\Tr\left(XX^*\right)\right)^{1/2}. \]

\subsection{Application to moments of partially transposed random states} \hfill\vspace{0.1cm}

\begin{lemma} \label{lem:Lipschitz-constant}
Given $k\in\mathbb N$, define $f_k:M_{d^2,s}(\mathbb C)\to\C$ by, for all $X\in M_{d^2,s}(\mathbb C)$,
\begin{equation} \label{eq:def-f}
f_k(X) := \Tr\left(\left(\left(XX^*\right)^\Gamma\right)^k\right). 
\end{equation}
$f_k$ satisfies, for all $X,Y\in M_{d^2,s}(\mathbb C)$,
\[ \left|f_k(X)-f_k(Y)\right| \leq 2kd\max\left(\left\|(XX^*)^\Gamma\right\|_\infty,\left\|(XX^*)^\Gamma\right\|_\infty\right)^{k-1} \max\left(\|X\|_\infty,\|Y\|_\infty\right) \|X-Y\|_2. \]
\end{lemma}

\begin{proof}
    Set $A=(XX^*)^\Gamma$ and $B=(YY^*)^\Gamma$. Observe that we can re-write $A^k-B^k$ as a telescoping sum, namely
    \[ A^k-B^k = \sum_{j=0}^{k-1} A^j(A-B)B^{k-1-j}. \]
    We thus get by the triangle inequality
    \[ \left|\Tr\big(A^k\big)-\Tr\big(B^k\big)\right| = \left|\Tr\left(\sum_{j=0}^{k-1} A^j(A-B)B^{k-1-j}\right)\right| \leq \sum_{j=0}^{k-1} \left|\Tr\left(A^j(A-B)B^{k-1-j}\right)\right| .\]
    Now, for each $0\leq j\leq k-1$, we have by H\"{o}lder's inequality
    \begin{align*}
        \left|\Tr\left(A^j(A-B)B^{k-1-j}\right)\right| & \leq \big\|A^j\big\|_\infty \big\|B^{k-1-j}\big\|_\infty \|A-B\|_1 \\
        & = \|A\|_\infty^j\|B\|_\infty^{k-1-j}\|A-B\|_1 \\
        & \leq \max\left(\|A\|_\infty,\|B\|_\infty\right)^{k-1}\|A-B\|_1.
    \end{align*} 
    And we therefore obtain
    \[ \left|\Tr\big(A^k\big)-\Tr\big(B^k\big)\right| \leq k\max\left(\|A\|_\infty,\|B\|_\infty\right)^{k-1}\|A-B\|_1. \]

    Next, using the fact that, for any $M\in M_{d^2}(\mathbb C)$, $\|M\|_1\leq d\|M\|_2$ and $\|M^\Gamma\|_2=\|M\|_2$, we get
    \[ \|A-B\|_1 \leq d\|A-B\|_2 = d\left\|(XX^*)^\Gamma-(YY^*)^\Gamma\right\|_2 = d\|XX^*-YY^*\|_2. \]
    And finally, observing that $XX^*-YY^*=X(X^*-Y^*)+(X-Y)Y^*$, we have by the triangle inequality and H\"{o}lder's inequality
    \begin{align*}
        \|XX^*-YY^*\|_2 & \leq \|X(X^*-Y^*)\|_2 + \|(X-Y)Y^*\|_2 \\ 
        & \leq \|X\|_\infty\|X^*-Y^*\|_2+\|Y^*\|_\infty\|X-Y\|_2 \\
        & \leq 2\max\left(\|X\|_\infty,\|Y\|_\infty\right)\|X-Y\|_2.
    \end{align*}
    Putting everything together yields precisely the announced result.
\end{proof}

Lemma \ref{lem:Lipschitz-constant} shows that the function $f_k$, as defined in equation \eqref{eq:def-f}, has a `small' Lipschitz constant (with respect to the Hilbert--Schmidt norm) on subsets of the unit sphere of $M_{d^2,s}(\mathbb C)$ (for the Hilbert--Schmidt norm) where $X\mapsto\|X\|_\infty$ and $X\mapsto\|(XX^*)^\Gamma\|_\infty$ are upper bounded by a `small' value. So we now turn to upper bounding the typical value of these two functions for $X\in M_{d^2,s}(\mathbb C)$ uniformly distributed on the unit sphere for the Hilbert--Schmidt norm.

\begin{proposition} \label{prop:large-set}
Let $X\in M_{d^2,\lambda d^2}(\C)$ be uniformly distributed on the unit sphere for the Hilbert--Schmidt norm. Then, for any $k\in\mathbb N$, as soon as $d\geq (2+2/\sqrt{\lambda})^k$, we have
\[ \Prob\left(\|X\|_\infty > \frac{2}{d^{1-1/k}}\right) \leq e^{-\lambda d^{2+2/k}}. \]
\end{proposition}

\begin{proof}
    We apply Theorem \ref{th:Levy} with $D=2\lambda d^4$ and $f=\|\cdot\|_\infty$, which is clearly $1$-Lipschitz. We get that, for all $\epsilon>0$,
    \begin{equation} \label{eq:deviation-infty}
    \Prob\left(\|X\|_\infty > \E\left(\|X\|_\infty\right)+\epsilon\right) \leq e^{-\lambda d^4\epsilon^2}. 
    \end{equation}
    The only thing that remains to be done is to estimate $\E(\|X\|_\infty)$. We know that $X$ can be written as $X=G/\|G\|_2$, where $G\in M_{d^2,\lambda d^2}(\C)$ is a matrix whose entries are independent complex Gaussians with mean $0$ and variance $1$. It is well-known that, for such Gaussian matrix $G$, $\|G\|_2$ and $G/\|G\|_2$ are independent, which implies that
    \[ \E\left(\|X\|_\infty\right) = \frac{\E\left(\|G\|_\infty\right)}{\E\left(\|G\|_2\right)}. \]
    Now, we have (see e.g.~\cite[Proposition 6.31]{AubrunSzarek2017} for the first estimate and \cite[Proposition A.1]{AubrunSzarek2017} for the second one)
    \[ \E\left(\|G\|_\infty\right) \leq \sqrt{d^2}+\sqrt{\lambda d^2} = (1+\sqrt{\lambda})d \quad \text{and} \quad \E\left(\|G\|_2\right) \geq \sqrt{\lambda d^4-1} \geq \frac{\sqrt{\lambda}}{2}d^2 . \]
    We thus have
    \[ \E\left(\|X\|_\infty\right) \leq \frac{2+2/\sqrt{\lambda}}{d}. \]
    Plugging this upper bound into the concentration inequality \eqref{eq:deviation-infty} and taking $\epsilon=1/d^{1-1/k}$ gives exactly the announced result, after observing that $(2+2/\sqrt{\lambda})/d\leq 1/d^{1-1/k}$ for $d\geq(2+2/\sqrt{\lambda})^k$.
\end{proof}

\begin{lemma} \label{lem:op-norm-W-Gamma}
    Fix $\lambda>0$ and let $W\sim\mathcal W_{d^2,\lambda d^2}$. Then, 
    \[ \E\left(\big\|W^\Gamma\big\|_\infty\right) \leq C_0\lambda^2 d^2, \]
    where $C_0>0$ is an absolute constant.
\end{lemma}

\begin{proof}
    We start from the following simple observation: for any $k\in\mathbb N$, to be determined later, we have by ordering of Schatten norms and by Jensen's inequality
    \begin{equation}  \label{eq:ub-op-norm} \E\left(\big\|W^\Gamma\big\|_\infty\right) \leq \E\left(\big\|W^\Gamma\big\|_k\right) \leq \left(\E\left(\Tr\left(\left(W^\Gamma\right)^k\right)\right)\right)^{1/k}. \end{equation}
    We are now going to derive a rough upper bound on the right hand side of the above inequality, valid for $k$ growing sufficiently slowly with $d$. We recall that Proposition \ref{prop:first_moment} gives
    \[ \E \left(\Tr\left(\left(W^\Gamma\right)^k\right)\right) = \sum_{\pi\in S_k} \lambda^{\#(\pi)}d^{\#(\gamma^{-1}\pi) + \#(\pi\gamma) + 2\#(\pi)}. \]
    As explained in detail in the proof of Theorem \ref{th:first_moment_asymptotic}, for any $\pi\in S_k$, $\#(\gamma^{-1}\pi)+\#(\pi)\leq k+1$ and $\#(\pi\gamma)+\#(\pi)\leq k+1$. For all $0\leq r\leq\lfloor k/2\rfloor$, we define the following subsets of $S_k$:
    \begin{align*} S_k^\leftarrow(r) & := \left\{\pi\in S_k : \#(\gamma^{-1}\pi)+\#(\pi)=k+1-2r\right\}, \\
    S_k^\rightarrow(r) & := \left\{\pi\in S_k : \#(\gamma\pi)+\#(\pi)=k+1-2r\right\}. \end{align*}
    We can thus rewrite
    \[ \E \left(\Tr\left(\left(W^\Gamma\right)^k\right)\right) = \sum_{r=0}^{\lfloor k/2\rfloor} \left(\sum_{r'=0}^r \left(\sum_{\pi\in S_k^\leftarrow(r')\cap S_k^\rightarrow(r-r')} \lambda^{\#(\pi)}\right)\right)d^{2k+2-2r}. \]
    Now, for any $0\leq r\leq\lfloor k/2\rfloor$ and $0\leq r'\leq r$, we clearly have
    \[ \sum_{\pi\in S_k^\leftarrow(r')\cap S_k^\rightarrow(r-r')} \lambda^{\#(\pi)} \leq \left|S_k^\leftarrow(r')\cap S_k^\rightarrow(r-r')\right|\lambda^k \leq \left|S_k^\leftarrow(r')\right|\lambda^k. \]
    And we know from \cite[Lemma 12]{Mon13} that, for any $0\leq r\leq\lfloor k/2\rfloor$,
    \[ \left|S_k^\leftarrow(r)\right| \leq 4^{k}k^{3r+1}. \]
    Hence, putting everything together, we get
    \begin{align*} 
    \E \left(\Tr\left(\left(W^\Gamma\right)^k\right)\right) \leq \left(4\lambda d^{2}\right)^kkd^2 \left(\sum_{r=0}^{\lfloor k/2\rfloor} \frac{1}{d^{2r}}\left(\sum_{r'=0}^\delta k^{3r'}\right)\right) \leq \left(4\lambda d^{2}\right)^kkd^2 \left(\sum_{r=0}^{\lfloor k/2\rfloor}(r+1)\left(\frac{k^3}{d^2}\right)^r\right) .
    \end{align*}
    We now just have to note that, if $k\leq d^{2/3}$, then $(k^3/d^2)^r\leq 1$ for all $0\leq r\leq\lfloor k/2\rfloor$, and if $k\geq 3$, then $(\lfloor k/2\rfloor+1)(\lfloor k/2\rfloor+2)/2\leq k^2$. We therefore have, for all $3\leq k\leq d^{2/3}$, 
    \[ \E \left(\Tr\left(\left(W^\Gamma\right)^k\right)\right) \leq \left(4\lambda d^{2}\right)^k k^3d^2. \]
    Inserting the above upper bound into inequality \eqref{eq:ub-op-norm}, with $k=d^{2/3}$, gives
    \[  \E\left(\big\|W^\Gamma\big\|_\infty\right) \leq 4\lambda d^{2}d^{4/d^{2/3}} \leq 4\times 5^{4/5^{2/3}}\lambda d^2, \]
    where the last inequality is because $\max\{d^{4/d^{2/3}},\,d\in\mathbb N\}$ is attained for $d=5$. 
\end{proof}

\begin{remark}
    We mention that we could have established a more precise version of Lemma \ref{lem:op-norm-W-Gamma} by using \cite[Proposition 7.1]{Aubrun12}. Indeed, working a bit on the latter (in a similar fashion to what we did in the proof of Lemma \ref{lem:op-norm-W-Gamma}), one can derive that
    \[ \E\left(\left\|\left(W-\mathrm{diag}(W)\right)^\Gamma\right\|\right) \underset{d\to\infty}{=} \left(2\sqrt{\lambda}+o(1)\right)d^2. \]
    Next, we just have to observe that each entry $W_{ii}$, $1\leq i\leq d^2$, of $\mathrm{diag}(W)$ is distributed as $\|g_i\|^2$, for $g_1,\ldots,g_{d^2}\in\C^{\lambda d^2}$ independent standard Gaussian vectors. This implies that
    \[ \E\left(\left\|\mathrm{diag}(W)^\Gamma\right\|\right) = \E\left(\left\|\mathrm{diag}(W)\right\|\right) = \E\left(\max_{1\leq i\leq d^2} \|g_i\|^2\right) \underset{d\to\infty}{=} \left(\lambda+o(1)\right)d^2. \]
    And combining these two estimates gives 
    \[ \E\left(\left\|W^\Gamma\right\|\right) \leq \E\left(\left\|\left(W-\mathrm{diag}(W)\right)^\Gamma\right\|\right) + \E\left(\left\|\mathrm{diag}(W)^\Gamma\right\|\right) \underset{d\to\infty}{=} \left(\lambda+2\sqrt{\lambda}+o(1)\right)d^2. \]
    However, we do not need such tight upper bound for our purposes, which is why we chose to go through a rougher but more direct argumentation.
\end{remark}

\begin{proposition} \label{prop:large-set'}
Let $X\in M_{d^2,\lambda d^2}(\C)$ be uniformly distributed on the unit sphere for the Hilbert--Schmidt norm. Then, for any $k\in\mathbb N$, as soon as $d\geq\max(2+2/\sqrt{\lambda},\sqrt{C_0})^{k}$, we have
\[ \Prob\left(\big\|(XX^*)^\Gamma\big\|_\infty > \frac{2}{d^{2-2/k}}\right) \leq 3e^{-\lambda d^{2+2/k}/32}. \]
\end{proposition}

\begin{proof}
    We apply Corollary \ref{cor:mean-median} with $D=2\lambda d^4$ and $f:X\mapsto\|(XX^*)^\Gamma\|_\infty$, which is upper bounded by $1$ on the unit sphere for the Hilbert--Schmidt norm. This is because, for all $X\in M_{d^2,\lambda d^2}(\C)$,
    \[ \big\|(XX^*)^\Gamma\big\|_\infty \leq \big\|(XX^*)^\Gamma\big\|_2 = \left\|XX^*\right\|_2 = \|X\|_4^2 \leq \|X\|_2^2. \] 
    We also define
    \[ \Omega_k = \left\{ X\in M_{d^2,\lambda d^2}(\C) : \|X\|_2=1 \text{ and } \|X\|_\infty\leq\frac{2}{d^{1-1/k}} \right\}. \]
    We know from Proposition \ref{prop:large-set} that $\Omega_k$ has measure at least $1-e^{-\lambda d^{2+2/k}}$, which is larger than $3/4$. What is more, $f$ is $4/d^{1-1/k}$-Lipschitz on $\Omega_k$. Indeed, for all $X,Y\in M_{d^2,\lambda d^2}(\C)$,
    \begin{align*} 
    \left| \big\|(XX^*)^\Gamma\big\|_\infty - \big\|(YY^*)^\Gamma\big\|_\infty \right| & \leq \big\|(XX^*)^\Gamma - (YY^*)^\Gamma\big\|_\infty \\
    & \leq \big\|(XX^*)^\Gamma - (YY^*)^\Gamma\big\|_2 \\
    & = \left\|XX^* - YY^*\right\|_2 \\
    & \leq 2\max\left(\|X\|_\infty,\|Y\|_\infty\right)\|X-Y\|_2.
    \end{align*}
    We thus get that, for all $\epsilon>4(e^{-\lambda d^{2+2/k}}+2\sqrt{2\pi/\lambda}/d^{3-1/k})$, hence a fortiori for all $\epsilon>16\sqrt{2\pi/\lambda}/d^{3-1/k}$,
    \[ \Prob\left(\big\|(XX^*)^\Gamma\big\|_\infty > \mathbb E\left(\big\|(XX^*)^\Gamma\big\|_\infty\right) + \epsilon \right) \leq e^{-\lambda d^{2+2/k}} + 2e^{-\lambda d^{6-2/k}\epsilon^2/32}. \]
    In particular, taking $\epsilon=1/d^{2-2/k}$, we get
    \begin{equation} \label{eq:deviation-infty'}
    \Prob\left(\big\|(XX^*)^\Gamma\big\|_\infty > \mathbb E\left(\big\|(XX^*)^\Gamma\big\|_\infty\right) + \frac{1}{d^{2-2/k}} \right) \leq 3e^{-\lambda d^{2+2/k}/32}.
    \end{equation}
    
    The only thing that remains to be done is to estimate $\E(\|(XX^*)^\Gamma\|_\infty)$. As in the proof of Proposition \ref{prop:large-set}, we write $X=G/\|G\|_2$, where $G\in M_{d^2,\lambda d^2}(\C)$ is a matrix whose entries are independent complex Gaussians with mean $0$ and variance $1$. We then have
    \[ \E\left(\big\|(XX^*)^\Gamma\big\|_\infty\right) = \frac{\E\left(\big\|(GG^*)^\Gamma\big\|_\infty\right)}{\E\left(\|G\|_2^2\right)} . \]
    Now, we have 
    \[ \E\left(\big\|(GG^*)^\Gamma\big\|_\infty\right) \leq C_0\lambda d^2 \quad \text{and} \quad \E\left(\|G\|_2^2\right) = \lambda d^4 , \]
    where the first estimate is by Lemma \ref{lem:op-norm-W-Gamma}. We thus have
    \[ \E\left(\big\|(XX^*)^\Gamma\big\|_\infty\right) \leq \frac{C_0}{d^2}. \]
    Plugging this upper bound into the concentration inequality \eqref{eq:deviation-infty'} gives exactly the announced result, after observing that $C_0/d^2\leq 1/d^{2-2/k}$ for $d\geq\sqrt{C_0}^{k}$.
\end{proof}

\begin{theorem} \label{th:pk_concentration_improved}
    Fix $k\in\mathbb N$ and define $f_k$ as in equation \eqref{eq:def-f}. Let $X\in M_{d^2,\lambda d^2}(\C)$ be uniformly distributed on the unit sphere for the Hilbert--Schmidt norm. Then, as soon as $d\geq C(\lambda,k)$, for every $\delta>C'(\lambda,k)/d^{1+1/k}$,
    \[ \Prob\left(\left|f_k(X)- \frac{M_k(\lambda)}{\lambda^k}\frac{1}{d^{2k-2}}\right| > \frac{\delta}{d^{2k-2}} \right) \leq e^{-c(\lambda,k)d^{2+2/k}\delta^2}, \]
    where $C(\lambda,k),C'(\lambda,k),c(\lambda,k)>0$ are constants depending only on $\lambda$ and $k$.
\end{theorem}

\begin{proof}
    We apply Corollary \ref{cor:mean-median} with $D=2\lambda d^4$ and $f_k:X\mapsto\Tr(((XX^*)^\Gamma)^k)$, which is upper bounded by $1$ on the unit sphere for the Hilbert--Schmidt norm. This is because, for all $X\in M_{d^2,\lambda d^2}(\C)$, for $k=1$ we have
    \[ \Tr\left((XX^*)^\Gamma\right) = \Tr\left(XX^*\right) = \|X\|_2^2, \]
    while for $k\geq 2$ we have
    \[ \Tr\left(\big((XX^*)^\Gamma\big)^k\right) \leq \big\|(XX^*)^\Gamma\big\|_k^k \leq \big\|(XX^*)^\Gamma\big\|_2^k = \left\|XX^*\right\|_2^k = \|X\|_4^{2k} \leq \|X\|_2^{2k}. \] 
    We also define
    \[ \Omega_k = \left\{ X\in M_{d^2,\lambda d^2}(\C) : \|X\|_2=1 \text{ and } \|X\|_\infty\leq\frac{2}{d^{1-1/k}},\ \big\|(XX^*)^\Gamma\big\|_\infty\leq\frac{2}{d^{2-2/k}} \right\}. \]
    We know from Propositions \ref{prop:large-set} and \ref{prop:large-set'}, together with the union bound, that $\Omega_k$ has measure at least $1-e^{-\lambda d^{2+2/k}}-3e^{-\lambda d^{2+2/k}/32}$, hence at least $1-4e^{-\lambda d^{2+2/k}/32}$, which is larger than $3/4$. What is more, we know from Lemma \ref{lem:Lipschitz-constant} that $f_k$ has a Lipschitz constant $2^{k+1}k/d^{2k-3+1/k}$ on $\Omega_k$. 

    We thus get that, for all $\epsilon>4(4e^{-\lambda d^{2+2/k}/32}+2^kk\sqrt{2\pi/\lambda}/d^{2k-1+1/k})$, hence a fortiori for all $\epsilon>2^{k+3}k\sqrt{2\pi/\lambda}/d^{2k-1+1/k}$,
    \[ \Prob\left(\left|f_k(X)- \mathbb E\left(f_k(X)\right)\right| > \epsilon \right) \leq 4e^{-\lambda d^{2+2/k}/32} + 2e^{-\lambda d^{4k-2+2/k}/2^{2k+5}k^2}. \]
    In particular, taking $\epsilon=\delta/d^{2k-2}$, we get that, for all $\delta>C_{\lambda,k}/d^{1+1/k}$,
    \begin{equation} \label{eq:deviation-infty''}
    \Prob\left(\left|f_k(X)- \mathbb E\left(f_k(X)\right)\right| > \frac{\delta}{d^{2k-2}} \right) \leq 6e^{-c_{\lambda,k}d^{2+2/k}\delta^2},
    \end{equation}
    where $C_{\lambda,k}=2^{k+3}k\sqrt{2\pi/\lambda}$ and $c_{\lambda,k}=\lambda/2^{2k+5}k^2$.
    
    The only thing that remains to be done is to estimate $\E(f_k(X))$. As in the proofs of Propositions \ref{prop:large-set} and \ref{prop:large-set'}, we write $X=G/\|G\|_2$, where $G\in M_{d^2,\lambda d^2}(\C)$ is a matrix whose entries are independent complex Gaussians with mean $0$ and variance $1$. We then have
    \[ \E\left(f_k(X)\right) = \frac{\E\left(f_k(G)\right)}{\E\left(\|G\|_2^{2k}\right)} . \]
    Now, we have 
    \[ \E\left(f_k(G)\right) = M_k(\lambda)d^{2k+2} + O\left(d^{2k}\right) \quad \text{and} \quad \E\left(\|G\|_2^{2k}\right) = \lambda^k d^{4k} + O\left(d^{4k-4}\right), \]
    where the first estimate is by Theorem \ref{th:first_moment_asymptotic}. We thus have
    \[ \E\left(f_k(X)\right) = \frac{M_k(\lambda)}{\lambda^k}\frac{1}{d^{2k-2}} + O\left(\frac{1}{d^{2k}}\right). \]
    Plugging this estimate into the concentration inequality \eqref{eq:deviation-infty''} gives that, for all $\delta>C_{\lambda,k}/d^{1+1/k}$ (which implies that $\delta/d^{2k-2}>C_{\lambda,k}/d^{2k-1+1/k}=O(1/d^{2k})$),
    \[  \Prob\left(\left|f_k(X)- \frac{M_k(\lambda)}{\lambda^k}\frac{1}{d^{2k-2}}\right| > \frac{2\delta}{d^{2k-2}} \right) \leq 6e^{-c_{\lambda,k}d^{2+2/k}\delta^2}, \]
    which is exactly the announced result, up to relabeling constants.
\end{proof}

\section{Recurrence identity for asymptotic average moment coefficients}\label{app:moment}

We recall the expression of the asymptotic average moment coefficients, introduced in
the main text: for all $n\in\mathbb N$,
\[ M_n(\lambda)=\sum_{l=0}^{\lfloor n/2\rfloor}
\binom{n}{2l}\mathrm{Cat}(l)\lambda^{n-l}. \]
The goal of this appendix is to prove the recurrence identity used in Lemma~\ref{lem:genC}, namely:
for all $n\geq 2$,
\begin{equation} \label{eq:identity-M_n}
  M_n(\lambda)-\lambda M_{n-1}(\lambda)
  =\lambda\sum_{p+q=n-2} M_p(\lambda)\,M_q(\lambda).
\end{equation}

Let us start by estimating the left-hand side of equation \eqref{eq:identity-M_n}. For
all $n\ge2$, we have
\[
  M_n(\lambda)-\lambda M_{n-1}(\lambda)
  = \sum_{l=0}^{\lfloor n/2\rfloor}
    \left[\binom{n}{2l}-\binom{n-1}{2l}\right]
    \mathrm{Cat}(l)\lambda^{n-l}.
\]
The $l=0$ term vanishes, and for $l\ge1$, Pascal's identity gives
$\binom{n}{2l}-\binom{n-1}{2l}=\binom{n-1}{2l-1}$. So re-indexing with
$L=l-1$, we get
\begin{equation}\label{eq:lhs_moment}
  M_n(\lambda)-\lambda M_{n-1}(\lambda)
  = \sum_{L=0}^{\lfloor n/2\rfloor-1}
    \binom{n-1}{2L+1}\,\mathrm{Cat}(L+1)\lambda^{n-1-L}.
\end{equation}

Let us now turn to estimating the right-hand side of equation \eqref{eq:identity-M_n}.
For all $n\geq 2$, inserting the definitions of $M_p(\lambda)$ and $M_{n-2-p}(\lambda)$
for all $0\leq p\leq n-2$ and re-arranging the sums, we have
\[ \sum_{p=0}^{n-2} M_p(\lambda)M_{n-2-p}(\lambda)
= \sum_{l,l'=0}^{\lfloor(n-2)/2\rfloor}
    \mathrm{Cat}(l)\mathrm{Cat}(l')
    \lambda^{n-2-(l+l')}
    \left(\sum_{p=2l}^{n-2-2l'}
    \binom{p}{2l}\binom{n-2-p}{2l'}\right). \]
We can compute the inner sum by one of the binomial identities associated to the
Chu--Vandermonde convolution (see e.g.~\cite{Gould} for a list and proofs): for all
$0\leq l,l'\leq\lfloor(n-2)/2\rfloor$,
\[
  \sum_{p=2\ell}^{n-2-2l'} \binom{p}{2l}\binom{n-2-p}{2l'}
  = \binom{n-1}{2l+2l'+1}.
\]
Setting $L=l+l'$ and using the Catalan recurrence identity
$\mathrm{Cat}(L+1)=\sum_{l=0}^L
\mathrm{Cat}(l)\mathrm{Cat}(L-l)$, we thus get
\begin{align}\label{eq:rhs_moment}
  \sum_{p=0}^{n-2} M_p(\lambda)M_q(\lambda) & = \sum_{L=0}^{\lfloor n/2\rfloor-1}
    \binom{n-1}{2L+1} \lambda^{n-1-L} \left(\sum_{l=0}^L\mathrm{Cat}(l)\mathrm{Cat}(L-l)\right) \nonumber \\
  & = \sum_{L=0}^{\lfloor n/2\rfloor-1}
    \binom{n-1}{2L+1}\mathrm{Cat}(L+1)\lambda^{n-1-L}.
\end{align}

Comparing equations \eqref{eq:lhs_moment} and \eqref{eq:rhs_moment} completes the proof of identity \eqref{eq:identity-M_n}.

\bibliographystyle{alpha}
\bibliography{references}

\end{document}